\documentclass[aps,pra,preprint,amsmath,amssymb,longbibliography]{revtex4-2}
\usepackage{graphicx}
\usepackage{bm,amsthm}
\usepackage{xcolor}
\usepackage{bbold}
\usepackage{natbib}
\usepackage{url}

\usepackage{tensor,pbox,hyperref}
\usepackage{braket}
\usepackage[qm]{qcircuit} 

\newcommand{\qc}[2][@C=.8em @R=.7em]{\Qcircuit #1 {#2}}

\newcommand{\rvline}{\hspace*{-\arraycolsep}\vline\hspace*{-\arraycolsep}}

\newcommand{\kett}[1]{\ket{#1}\rangle}

\newcommand{\bea}{\begin{eqnarray}}
\newcommand{\ea}{\end{eqnarray}}
\newcommand{\eea}{\end{eqnarray}}
\newcommand{\ord}{\mathcal{O}}
\newcommand{\Nb}{{N_\text{bit}}}

\newcommand{\nb}{{n_\text{bit}}}
\newcommand{\Nl}{{N_\text{smpl}}}

\newcommand{\rI}{{\text{I}}}
\newcommand{\rII}{{\text{II}}}

\newcommand{\defas}[0]{\mathrel{\mathop:}=} 
\newcommand{\asdef}[0]{=\mathrel{\mathop:}} 

\usepackage{accents}
\newcommand{\dbtilde}[1]{\accentset{\approx}{#1}}
\DeclareRobustCommand{\ptil}{\dbtilde{p}} 

\newtheorem{theorem}{Theorem}[section]
\newtheorem{definition}[theorem]{Definition}

\newtheorem{lemma}[theorem]{Lemma}

\begin{document}
\newcommand{\ri}{ i}
\newcommand{\re}{ e}
\newcommand{\bx}{{\bm x}}
\newcommand{\bd}{{\bm d}}
\newcommand{\be}{{\bm e}}
\newcommand{\br}{{\bm r}}
\newcommand{\bk}{{\bm k}}
\newcommand{\bA}{{\bm A}}
\newcommand{\bD}{{\bm D}}
\newcommand{\bE}{{\bm E}}
\newcommand{\bB}{{\bm B}}
\newcommand{\bI}{{\bm I}}
\newcommand{\bH}{{\bm H}}
\newcommand{\bL}{{\bm L}}
\newcommand{\bR}{{\bm R}}
\newcommand{\bZero}{{\bm 0}}
\newcommand{\bM}{{\bm M}}
\newcommand{\bX}{{\bm X}}
\newcommand{\bn}{{\bm n}}
\newcommand{\bs}{{\bm s}}
\newcommand{\bv}{{\bm v}}
\newcommand{\tbs}{\tilde{\bm s}}
\newcommand{\tV}{\tilde{V}}
\newcommand{\rSi}{{\rm Si}}
\newcommand{\beps}{\mbox{\boldmath{$\epsilon$}}}
\newcommand{\bGamma}{\mbox{\boldmath{$\Gamma$}}}
\newcommand{\bxi}{\mbox{\boldmath{$\xi$}}}
\newcommand{\rg}{{\rm g}}
\newcommand{\tr}{{\rm tr}}
\newcommand{\xmax}{x_{\rm max}}
\newcommand{\xb}{\overline{x}}
\newcommand{\pb}{\overline{p}}
\newcommand{\ra}{{\rm a}}
\newcommand{\rx}{{\rm x}}
\newcommand{\rs}{{\rm s}}
\newcommand{\rP}{{\rm P}}
\newcommand{\up}{\uparrow}
\newcommand{\down}{\downarrow}
\newcommand{\hc}{H_{\rm cond}}
\newcommand{\kb}{k_{\rm B}}
\newcommand{\cI}{\mathcal{I}}
\newcommand{\tit}{\tilde{t}}
\newcommand{\cE}{\mathcal{E}}
\newcommand{\cC}{\mathcal{C}}
\newcommand{\Ubs}{U_{\rm BS}}
\newcommand{\sech}{{\rm sech}}
\newcommand{\xs}{{x_1,\ldots,x_N}}
\newcommand{\qq}{{\bf ???}}
\newcommand*{\etal}{\textit{et al.}}
\newcommand{\empt}{\mbox{ }}
\newcommand{\tinI}{{\text{\tiny 1}}}
\newcommand{\tinDrei}{{\text{\tiny 3}}}
\newcommand{\dbla}{{\text{\tiny 2,3}}}
\newcommand{\dbll}{{\text{\tiny 1,2}}}
\newcommand{\dble}{{\tiny\mbox{1,2}}}
\newcommand{\trpl}{{\tiny\mbox{1,2,3}}}
\newcommand{\lra}{\leftrightarrow}
\newcommand{\rf}{{\text{rf}}}
\newcommand{\ttil}{\tilde{t}}
\newcommand{\Vtil}{\tilde{V}}
\newcommand{\Stil}{\tilde{S}}
\newcommand{\sign}{\text{sign}}
\newcommand{\abs}{\text{abs}} 
\def\vec#1{\bm{#1}}
\def\ket#1{|#1\rangle}
\def\bra#1{\langle#1|}
\def\keps{\bm{k}\boldsymbol{\varepsilon}}
\def\dm{\boldsymbol{\wp}}


\title{The Twin-World road to reality in quantum mechanics}
\author{Daniel Braun}
\affiliation{Eberhard-Karls-Universit\"at T\"ubingen, Institut f\"ur Theoretische Physik, 72076 T\"ubingen, Germany}

\centerline{\today}
\begin{abstract}
I introduce a novel realistic, stochastic approach to quantum
mechanics by extending the recently proposed grabit formalism
\cite{braun_stochastic_2022} to two Twin Worlds.  According to the
picture developed, we live at the intersection of two worlds with
identical stochastic laws of evolution.  Our World is limited to that
intersection, and only coincidence events
from the two Twin Worlds, post-selected automatically by our
restriction to the intersection, have physical reality in Our World.  This
fully reproduces standard non-relativistic quantum mechanics,
including Born's rule and the violation of Bell's inequality. I derive
the stochastic evolution equation in each Twin World that fully
reproduces Schr\"odinger's equation for an arbitrary number of
particles with arbitrary interactions, and
demonstrate that hallmark quantum effects such as tunneling are
correctly reproduced.
\end{abstract}
\maketitle

\section{Introduction}
Almost 100 years after the formalization of quantum mechanics, there
is still no agreement on what the quantum mechanical wave function
actually represents.  Many different interpretations have been put
forward, starting from Schr\"odinger's own original ideas of some
smeared-out matter distribution, over the standard ``Copenhagen''
interpretation, the Bohm-de Broglie pilot wave, qbism, and many more
\cite{tumulka_foundations_nodate}.  All interpretations appear to agree, however, on
a statistical interpretation according to which the absolute value
squared of the wave function $|\psi(x)|^2$ represents the probability (or
probability density in the continuum case) to find a certain outcome $x$
if one performs a measurement of the observable $x$. This so-called
``Born-rule'' is a corner stone of the statistical
interpretation. Gleason's theorem and its generalizations show that
this the only choice for von Neumann measurements (with observables
represented by hermitian matrices) \cite{gleason_measures_nodate} or even POVM (=positive operator-valued
measure) measurements \cite{busch_quantum_2003}. 
Over the years, it has been attempted several times to construct
theories of the quantum world that follow classical statistical
mechanics, where variables have actual values that are only revealed
by the measurement process, and the quantum probabilities only reflect
our ignorance of those actual values that are distributed
according to a probability distribution.  That physical theories
should make statements about variables that actually do have a
well-defined value, whether observed or not, was a strong conviction
of Einstein \cite{Einstein35}.  That a theory turns out to be statistical rather than
deterministic was more acceptable for him, in spite of his often-cited quote
that God does not play dice \cite{Schilpp1959-SCHAEP-19,BornBio}.\\
However, with the advent of Bell's inequality and its various
generalizations that result from assuming a local realistic theory
with hidden variables as
well as the experimental verification of their violation, the
widespread agreement has become that Einstein was wrong and the case
has been settled. The standard paradigm is that the wave-function is all
we can possibly know about a quantum system.  The advent of quantum
information theory has strengthened that position even further, but
the discussions about the correct interpretation of the quantum
wavefunction, the ``measurement problem'', and of what happens during the collapse of the wave
function, continue.  ``Collapse models'' \`a la Ghirardi-Rimini-Weber
\cite{GhirardiRiminiWeber1986} or Diosi-Penrose
\cite{diosi_minimum_1989}, more and more constrained by
 experiments \cite{bassi_collapse_2023}, introduce
modifications of the Schr\"odinger equation that lead to a rapid
localization of the wave-function for macroscopic objects but retain
quantum superpositions of microscopic ones.   Bohmian
mechanics \cite{Bohm52,teufel_bohmian_2009,tumulka_foundations_nodate} proposes
particle positions guided by the wave function as realistic variables.
 The trajectories depend on the gauge choice for the wave function,
 and the latter makes the theory non-local. \\ 

Motivated originally by questions of simulatability of quantum
circuits, in \cite{braun_stochastic_2022} we asked the question
whether there are quantities other than the quantum mechanical wave
function that share the remarkable properties that  are exploited
artfully in quantum algorithms.  These are the possibility
of interference in an exponentially large Hilbert space, combined with
the possibility of local manipulation of the state by acting on only
one or two qubits at the same time, and physical reality of the
quantity.  Surprisingly, we found that  
derivatives of a probability distribution share these properties.
More precisely, one can define a ``grabit'' state (for ``gradient
bit'') as in \eqref{eq:psi} below, with one partial derivative for
each particle (or qubit).  We then showed that one can find a set of
universal stochastic matrices by replacing each universal quantum gate
with a corresponding specific stochastic matrix. The concatenation of
these stochastic matrices leads to an automated 
translation of any quantum circuit to a classical stochastic process
such that the corresponding grabit state follows, up to a prefactor,
exactly the quantum mechanical evolution of the realified quantum
state (i.e.~the original complex quantum state split into real- and
imaginary parts).

There were two caveats of the approach, however: i.)  the prefactor decays
typically exponentially with the number of interference-generating
gates; and ii.)
contrary to standard quantum mechanics the probabilities for the
final outcomes are given by the absolute values of the amplitudes of
the grabit state, not the amplitude squared, a rule called ``Born-1''
rule in \cite{braun_stochastic_2022}. The first caveat could be
remedied to large extent by ``refreshment gates'', to be
implemented after each interference-generating gate. They are
non-linear gates based on an estimate of the quantum state (estimated ideally
from the probability distribution, in emulations from a
histogram of possible outcomes) that make the grabit-state
interference-free.  In the limit of an infinitely large number of
samples, the evolution of the 
grabit state  follows exactly the true quantum mechanical,
realified state, without decaying. In actual emulations of quantum
algorithms that produce a small number of output peaks compared to the
exponentially many possible outputs
(such as Shor's factoring algorithm), it
turned out, however, that the number of samples required for
sufficiently precise estimates for the refreshment gates to work will
typically be exponentially large. This limits the practical
usefulness of the stochastic emulation for such algorithms, even
though for other quantum algorithms, such as quantum approximate optimization algorithm (QAOA) at
level $p=1$, it turned out to have
sampling complexity comparable to an ideal quantum computer \cite{braun_stochastic_2022}.\\

In the present work I show that quantum mechanics can be recovered
exactly by identifying the quantum mechanical wave function with  the
derivative of a probability distribution. This naturally leads to a
Twin World interpretation of the quantum world, where all that happens
in Our World are post-selected coincidence events from the two Twin
Worlds.  The Twin World picture offers a physical explanation of the origin of the standard Born
rule (called ``Born-2'' rule in \cite{braun_stochastic_2022} and in
the following for disambiguation), and gives rise to a realistic, stochastic theory of the quantum world. 
In the following I briefly review the grabit formalism, then show that
enforcement of the Born-2 rule leads naturally to a Twin World
interpretation of quantum reality and the exact same violation of
the CHSH inequality as in standard quantum mechanics (QM).  I
then explore dynamics and construct a 
stochastic evolution equation that correctly reproduces the
Schr\"odinger equation of a particle on a 1D line, including hallmark 
effects of quantum mechanics such as the tunneling effect, before
extending the theory to any dimension and the many-particle case. 

The
formalism developed allows a microscopic realistic picture of the
quantum world, with quantum mechanics indeed just a statistical theory
on top, albeit a less trivial one than what is usually assumed for
coding a local realistic theory that gives rise to Bell's
inequality and its variants. Rather than excluding local realistic
theories \emph{per se}, it becomes clear that Bell's inequality is built on a
{\em specific probabilistic model} of locality and realism, such that
its experimentally established violation invalidates that specific
model, but not necessarily {\em any} local realistic theory.  The
model proposed here is explicitly realistic, and since it reproduces
perfectly the standard quantum statistics it can, of course, not be
local in the sense of the usual specific probabilistic model going
back to Bell. But it is local in the same sense in which quantum
mechanics peacefully co-exists with special relativity, i.e.~that, as
far as we know, 
it does not admit superluminal information transfer between
measurement settings  on Alice's side to probability distributions 
of measurement results on Bob's side.  Since we now have 
realistic variables such as particle positions rather than only the
wave function to describe physical reality, this is the only locality
that really counts. The measurement problem is reduced to its
classical counterpart.

\section{The grabit formalism}\label{sec.grafor}
In standard QM, all we can know about the system is contained in a
state vector $\Psi$ in a Hilbert space over complex numbers.  Its
projections onto basis states are interpreted as probability
amplitudes to find a corresponding outcome of an experiment.  E.g., if
we have a system of $N$ particles in three dimensions, whose positions
are denoted $\bm x_1,\ldots,\bm x_N$ with $\bm x_k\in \mathbb{R}^3$, the wave
function $\Psi(\bm x_1,\ldots,\bm x_N)\in \mathbb{C}$ has, according
to the Born-2 rule, the physical meaning that $|\Psi(\bm x_1,\ldots,\bm
x_N)|^2$ denotes the probability to jointly find the particles  in infinitesimal volume elements at
positions $\bm x_i$ for $i\in\{1,\ldots,N\}$. Exactly the same information is, of course,
contained in a ``realified'' version, $\Phi(\bm x_1,\ldots,\bm
x_N)\equiv (\Re \Psi(\bm x_1,\ldots,\bm x_N),\Im \Psi(\bm
x_1,\ldots,\bm x_N))$. Also the propagation of the state via a unitary
matrix can be easily split into real and imaginary parts. 
In more detail, for any system with finite dimensional Hilbert space, any evolution with a complex unitary matrix $U$ with matrix elements $U_{ij}$ in an orthonormal basis that propagates a complex quantum state $\Psi$,  $\Psi'_i=U_{ij}\Psi_j$ (with summation convention
over pairs of identical indices), can be rewritten equivalently as  
${\Phi}'_i=\tilde{U}_{ij}{\Phi}_j$ where each 
complex matrix element $U_{ij}$ is replaced by a real 2 $\times$ 2
matrix \cite{aharonov_simple_2003},
\begin{equation}
  \label{eq:Util}
  U_{ij}\mapsto
  (\tilde{U}_{i\nu,j\mu})_{\nu,\mu\in\{0,1\}}=\begin{pmatrix}
    \Re U_{ij}&-\Im U_{ij}\\
    \Im U_{ij}&\Re U_{ij}
    \end{pmatrix}\,.
  \end{equation}

  For a system of $\nb$ qubits, real and imaginary part of a
  multi-qubit state are distinguished by a single extra (real) qubit, called
  the ReIm qubit.  $\nb$ complex qubits are completely equivalent to
  $\Nb=\nb+1$ real qubits.\\ 
\begin{table}
  \centering
\begin{tabular}{p{2cm}||p{2cm}|p{2cm}|p{2cm}|p{2cm}|p{3.5cm}|p{3.5cm}|}
    &complex state& realified state & b4v or ppv probability & relative
                                                        frequency &
                                                                    physical
                                                                    probability
 in each Twin World& coincidence probability in Our World  \\\hline
  QM&$\Psi$& $\Phi$ & & & &  \\\hline
   Twin World&$\psi$& $\phi$ & $P$ &$R$ &$\tilde{p}^{\rI,\rII}$& $\ptil$ \\\hline
  \end{tabular}\label{tab0}
  \caption{Overview of symbols. Here, b4v=''byte-four value'',
    $I\equiv (i,\sigma)$ where $i,\sigma\in\{0,1\}$, and
    ppv=''phased position variable'' $(\rho,\sigma, \bm x)$ where
    $\rho,\sigma\in\{0,1\}$ and $x\in \mathbb{Z}_N$ for a single particle
    on a 1D line with $N$ sites (see section \ref{sec.dynami}).  } 
\end{table}

The central idea of the grabit approach is to represent the realified
$\Phi(\bm x_1,\ldots,\bm x_N)$ as ($N$-th) derivative of a certain
probability distribution $P(\bm x_1,\ldots,\bm x_N)$.  In the simplest
case, where the particles are considered in 1D, one could, e.g.,
define 
\begin{equation}
  \label{eq:psi}
\phi_\rho(x_1,\ldots,x_N)=\partial_{x_1}\ldots \partial_{x_N}P_\rho(x_1,\ldots,x_N) \,,
\end{equation}
where $\rho=0,1$ labels the real- and imaginary part of $\psi$, respectively.
Owing to the fact that this object can be positive or negative, it is
capable of full interference in a tensor-product Hilbert space, just
as the quantum mechanical wave function, and has, at the same time, a
very clear and operational physical meaning.   For
discrete variables, the derivatives can be replaced by differences,
which leads to the representation of each realified quantum mechanical qubit by
a ``grabit'' (short for ``gradient bit''), consisting of two classical
stochastic bits and an additional bit that codes the information what is the real and what is the imaginary part.   The two classical bits code four ``byte-four
values'' (b4vs) $I\equiv (i,\sigma)$ where $i\in\{0,1\}$ codes the
``byte-logical value'', and $\sigma\in\{0,1\}$ the ``gradient value''
that determines the sign of the difference of the probabilities.

There is some freedom in mapping complex states and
unitaries to real ones, and additional freedom in going from realified
quantum states to their grabit representation.  Adding one global real
qubit with pure orthogonal basis states $\ket{\rho}$ as outlined above is one possibility.
Alternatively, one may realify the local computational 
  bases, in which case one has to add a real ReIm bit with state
  labels $\rho_\nu$ for each
  subsystem.
The two labels $\sigma_\nu,\rho_\nu$ then determine locally
  the four possible phase factors $1,i,-1,-i$ for $\sigma_\nu,\rho_\nu\in
  \{00,01,10,11\}$, respectively. 
  Combining these factors with 
  probabilities in $[0,1]$ one can create any phase in $[0,2\pi[$ for
  each subsystem.  Explicitly, the local phase factors are
  $(-1)^{\sigma_\nu}(i)^{\rho_\nu}$. This coding requires three classical
  stochastic bits  for each complex quantum bit, and corresponds to
  the situation in standard quantum mechanics, where 
phase factors from subsystems combine to a global one,
  $(-1)^{\sigma}(i)^{\rho}$, with
  $\sigma=\sum_{\nu=1}^{\nb}\sigma_\nu$ and
  $\rho=\sum_{\nu=1}^{\nb}\rho_\nu$.  Explicitly, the complex
  amplitude $\psi_{\rho,\bm i}$ is given by
  \begin{equation}
    \label{eq:psirhoi}
    \psi_{\rho,\bm i}=(i)^\rho\sum_{\bm \rho|\rho=\sum_{\nu=1}^\nb
      \rho_\nu}\sum_{\bm
      \sigma}(-1)^{\sum_{\nu=1}^\nb\sigma_\nu}P_{\bm \rho,2\bm i+\bm \sigma}\,.
  \end{equation}
 Local stochastic operations can then be performed in parallel. 
  The alternative, with one
global ReIm grabit, used in \cite{braun_stochastic_2022} and outlined
at the beginning of this Section, leads to a representation of an
$\Nb$-grabit state with $\Nb=\nb+1$ given by 
 \begin{equation}
  \label{eq:psiN}
  \phi_{\bm i}=\sum_{\bm \sigma
    \in\{0,1\}^{\otimes \Nb}}(-1)^{\sum_{i=1}^\Nb\sigma_i}\,P_{2\bm i+\bm \sigma}\,.
\end{equation}
It is a vector in $\mathbb{R}^{2^\Nb}$, with elements labeled by the
index $\bm i=(i_1,\ldots,i_\Nb)$, where the ReIm bit was absorbed into
$\bm i$, e.g.~as highest or lowest significant bit.  
As was shown in \cite{braun_stochastic_2022}, any pure quantum
mechanical state $\Psi\in\mathbb{C}^{2^\nb}$ can be represented this way, up to
a real prefactor $a\ne 0$.  So here we need $2(\nb+1)$ classical
stochastic bits.  In fact,  even $2\nb+1$ many are enough, as for the
ReIm information a classical bit $\rho$ suffices. This version reduces
the memory requirement but prevents executing several gates that change the
global ReIm in parallel. Loops over random realizations can still be
parallelized, though.

Even more memory-efficient would be to have just a global sign bit as
well, 
  \begin{eqnarray}
    \label{eq:psirhoiG}
    \psi_{\rho,\bm i}&=&(i)^\rho\sum_{      \sigma}(-1)^\sigma P_{    \rho,\sigma,\bm i}\\
        \Leftrightarrow \phi_{\rho,\bm i}&=&\sum_{      \sigma}(-1)^\sigma P_{ \rho,\sigma,\bm i}\,.
  \end{eqnarray}
This removes the possibility of local manipulation of phases and hence
requires all gates, even real single-qubit gates to interact with the
global sign bit, but since
phases of any component in a superposition of states of a multi-particle system are only
defined up to a global phase factor, this still allows to
represent any pure quantum state,
\begin{equation}
  \label{eq:psiSimpl}
  \ket{\psi}=\sum_{\rho,\bm i}\psi_{\rho,\bm i}\ket{\bm i}=\sum_{\bm
    i}(\phi_{0,\bm i}+i \phi_{1,\bm i})\ket{\bm i}\,,
\end{equation}
and with a cost of $\nb+2$ classical stochastic bits, this seems to be
the most economical representation of a complex quantum state of $\nb$
qubits. \\

In the following, in order to stay close to standard
  quantum mechanics and also allow local interference, independent of
  the other subsystems, we largely stick to the convention of
  \cite{braun_stochastic_2022} with a local gradient value $\sigma_i$ for
  each subsystem, but a global ReIm bit $\rho$. Any unitary transformation on $\mathbb{C}^{2^\nb}$ can then be represented by a stochastic map that propagates the probability distribution  
$P_I$ (where $I\equiv (2\bm i+\bm\sigma)=0,\ldots, 4^\Nb-1$), such that
the final quantum mechanical state can be extracted from the final
probability distribution according to \eqref{eq:psiN}, again up to a
prefactor.  Indeed, each single-qubit or two-qubit unitary
transformation from a universal gate set can be translated to a
corresponding stochastic map acting on at most three grabits, such that there is an automated
translation of any quantum circuit to a corresponding grabit quantum
circuit, i.e.~a classical stochastic map, that propagates the state $\psi$ coded via \eqref{eq:psiN} up
to a prefactor exactly as the quantum circuit (see Theorem 3.3 in
\cite{braun_stochastic_2022} for a precise formulation).

\section{The Twin-World interpretation of quantum mechanics}\label{sec.TwiWor}
As mentioned, the
prefactor of the state $\phi$ decays typically exponentially with the number of
interference-generating gates.  These are gates such as the Hadamard
gate that create quantum-mechanical superpositions in the
computational basis. The problem is, generally speaking,
that destructive interference that manifests itself in a reduction of
a component of $\phi_{\bm i}$, arises from a difference of at least
two probabilities over b4vs. E.g.~after implementing a squared
Hadamard gate acting on $\ket{0}$, QM gives $H^2\ket{0}=\ket{0}$. With
a grabit emulation, and skipping momentarily the ReIm grabit as the
states considered here are real, after the two Hadamard gates the grabit has the
probability distribution $P=(2,0,1,1)/4$ over b4vs and leads to
$\phi=(1,0)/2$. While this is the correct state up to the prefactor
$1/2$, the physical probability distribution $\tilde{p}_i$
marginalized over the gradient variable $\sigma$, gives
$\tilde{p}_i\equiv \sum_{\sigma=0,1}P_{2i+\sigma}=1/2$, i.e.~there is no contrast $\tilde{p}_0-\tilde{p}_1$ at
all. And with every further iteration of $H^2$, the prefactor of 
$\phi$ decreases by another factor 1/2 \cite{braun_stochastic_2022}. This has the practical
consequence that the number of samples that effectively contribute to
$\phi$ decreases exponentially with the number of
interference-generating gates.\\

Both of these problems can be
remedied at least partly by ``refreshments''.   These are non-linear
gates that re-enhance the contrast of the grabit states.  They are
defined formally in Section \ref{sec.dynami}, see eq.\eqref{eq:P''},
but in a  
nutshell and the present context of emulation of quantum algorithms,
they work as follows: After each interference-generating  
gate, estimate the state $\phi$ from the available samples,
i.e.~calculate an estimator $\hat\phi$ from the histogram $R_{2\bm
  i+\bm \sigma}$ of b4vs, in complete  
analogy of \eqref{eq:psi}, 
 \begin{equation}
  \label{eq:hatpsiN}
  \hat{\phi}_{\bm i}=\sum_{\bm \sigma
    \in\{0,1\}^{\otimes \Nb}}(-1)^{\sum_{i=1}^\Nb\sigma_i}\,R_{2\bm i+\bm \sigma}\,.
\end{equation}
Then redistribute the samples such that firstly for each $\bm i$ only a single
$\bm \sigma$, corresponding to the sign of $\hat\phi_{\bm i}$, contributes, and
secondly that for the state $\phi'$ after the
refreshment, $\phi_{\bm i}'\simeq a\, \hat\phi_{\bm i}$ with some
proportionality constant $a\geq 1$, as closely as possible.   In 
the limit of infinitely many samples and a finite Hilbert space, the agreement
can be made arbitrarily precise.  One then shows (see theorem 4.1
in \cite{braun_stochastic_2022}) that including the
refreshments allows one to end up with a physical probability
distribution given by the true realified quantum state, up to different
normalization,
\begin{equation}
  \label{eq:born1}
\tilde{p}_{\bm i}=|{\Phi}_{\bm i}|/||\, {\Phi}\, ||_1\,.    
\end{equation}
In practice, the number of realizations needed might still grow
exponentially for a state spreading over exponentially many basis states, but
for shallow quantum circuits the resulting classical stochastic simulation can
be of comparable efficiency as the actual quantum algorithm run on a quantum
computer, once the sampling complexity needed also for the latter is
taken into account.  

The second remaining caveat is that \eqref{eq:psiN} corresponds to
the ``Born-1 rule'', i.e.~the
physical probabilities are given by the absolute value of the grabit state
amplitudes rather than the squared absolute value as in the standard Born
rule.  This is of no
major concern for a quantum algorithm that ends in a highly peaked
state, where the peak position codes the result of the calculation.
Even if there is a relatively sparse set of final peak positions this is
not a problem, as  peaks remain peaks if one takes the square
root of the exact quantum mechanical distribution.  Hence, the
stochastic emulation will also end up with high probability in the same final bins as the
quantum algorithm, even if the probabilities themselves can be different from the quantum mechanical ones.\\ 

However, if one wants to entertain
the thought that maybe the realified quantum mechanical wavefunction
in Nature really is --- or at least 
can be consistently described as --- the derivative 
of a probability distribution, clearly one needs to go beyond the
Born-1 rule.  Interestingly now, if we tentatively accept \eqref{eq:psi} as origin
of a quantum state 
and the mechanism of refreshments (leading ultimately to
\eqref{eq:born1}), insisting on the validity of the Born-2 rule
largely {\em dictates} the physical reality of a {\em Twin World}:
There has to exist an exact copy of a first world where events
take place according to the same laws, but {\em only post-selected coincident events
from the two Twin Worlds have physical reality in Our
World}. Indeed,
let ${\ptil}_{\bm i}$ denote the physical probability to find event
$\bm i$ in Our World (e.g.~the event to find a particle at
a position $\bm x$ indexed by $\bm i$ in Our World). We have
\begin{eqnarray}
  \ptil_{\bm i}&=&|\Psi_{\bm i}|^2\label{eq:born2.5}\\
&=&  \sum_\rho\Phi_{\bm i,\rho}^2\label{eq:born2.4}\\                 &=& \sum_\rho\Phi_{\bm i,\rho}^2/\sum_{\bm i,\rho}
                             \Phi_{\bm i,\rho}^2\label{eq:born2.3}\\
&=& \sum_\rho|\phi_{\bm i,\rho}|^2/\sum_{\bm i,\rho} |\phi_{\bm i,\rho}|^2  \label{eq:born2.2}\\
&=&\sum_{\rho}\tilde{p}_{\bm i,\rho}\cdot
                             \tilde{p}_{\bm i,\rho}/\sum_{\bm
                             i,\rho}\tilde{p}_{\bm i,\rho}^2\label{eq:born2.1}
\end{eqnarray}
Eq.\eqref{eq:born2.5} is the usual Born-2 rule.  In \eqref{eq:born2.4} we have split into real and imaginary part. In \eqref{eq:born2.3} we used the normalization to $||\Phi||_2=1$ for the true quantum
mechanical state. To get to \eqref{eq:born2.2}, we have made use of
the proportionality of the grabit state to the true, realified quantum mechanical
wave-function in both Twin Worlds, eq.\eqref{eq:born1}.   The proportionality factor drops out here, so that the
different normalization does not matter. Eq.\eqref{eq:born2.1} follows since after a refreshment, the
probabilities are interference-free and given by the Born-1 rule, i.e. $\tilde{p}_{\bm
  i,\rho}=|\phi_{\bm i,\rho}|$.
Equation \eqref{eq:born2.1} has the direct interpretation of a 
probability of post-selected coincidence events $\bm i$ from the two Twin Worlds, marginalized
over the global ReIm index $\rho$.

Hence, we see that the Born-2 rule 
can be interpreted as the coincidence probability that two
independent, identical statistical ensembles give rise to the same outcome $\bm
i$. These two idependent, identical statistical ensembles present two
{\em Twin Worlds}:  Exactly the same stochastic dynamics happens in both of
them, and they do not interact.  All that happens in ``Our World''
consists then of post-selected coincidence events in the two Twin
Worlds.  By definition we don't see any other 
events, and this leads to an automatic post-selection of the
coincidence events.

To summarize, there are two worlds because of the Born-2 rule.  They
are Twin Worlds, i.e.~evolve under the exact same statistical laws, because the absolute value
square in the Born-2 rule consists of two identical factors.  The fact
that we have a product of probabilities means that the two Twin Worlds
are statistically independent, which excludes non-trivial interactions
between them that would correlate them. Maintaining normalization of the 
probabilities in Our World
implies post-selection of coincident events from the two Twin Worlds.  That
each factor can even be interpreted as a probability itself but at the
same time represents the absolute value of a quantum
mechanical probability amplitude requires, in the frame work of the
gradient states,  an interference-free
gradient state, i.e.~for each $\bm i,\rho$, there is at most one value of the
gradient variable $\bm \sigma$ populated. Hence, refreshments,
possibly after each interference-generating step, but at the latest
before each coincidence event are a necessary part of the dynamics in
each Twin World. 
\\

What picture of reality do we get from this?
If we accept that we live in a 3D configuration space, where each
particle has position coordinates $\bm x\in\mathbb{R}^3$, then the
most economical reality appears to require one additional coordinate
for the gradient direction. In principle, also one
of the components of $\bm x$ might be used, but this then begs the
question, which coordinate would be different from the others, and
immediately breaks isotropy of space.  If we accept one extra variable
for the gradient direction, we have a real vector $\bm X^w=(\bm
x^w,\sigma^w)\in\mathbb{R}^4$ for 
position and gradient variable with $w\in\{\rI,\rII\}$ for the two Twin
Worlds, i.e.~an 8D vector $\bm X=(\bm X^\rI,\bm X^\rII)$ that fully and realistically
describes the position of a particle in the Full World.  With ``Full
World'' I mean the two Twin Worlds including the gradient
variable. Contrary to previous propositions that
space might have extra dimensions, the extra
spatial dimensions in the Twin World II are not curled up to a very
small length scale, because otherwise they would
be curled up in Twin World I as well, as the two
Twin Worlds are identical twins.  But then they would be curled up in Our
World with $\bm x^\rI=\bm x^\rII$ which is clearly not the case.  This
leads to the question why we would not see particles disappear into
the higher dimensions, an argument that motivated the idea of curling
up in previous theories of higher dimensions \cite{clifton_modified_2012,emparan_black_2008,stojkovic_vanishing_2014,elvang_quantum_2014}.  I come back to this
question in the Discussion Section \ref{sec.discus}.

The dimensions corresponding to the gradient directions
$\sigma^\rI,\sigma^\rII$ might correspond to an internal degree of
freedom, similarly to a spin degree-of-freedom in standard QM.  As the
difference between two probability vectors already does the job, a simple
pseudo-spin-1/2, $\sigma^\rI,\sigma^\rII \in\{0,1\}\asdef
\mathbb Z_2$ suffices, without
the need of any length-scale or continuous variable attached to it.
Going back to complex states, the index $\rho$ takes only values 0,1
anyhow, and might correspond to an additional internal, spin-like
degree of freedom as well.  Each Twin World can carry its own ReIm bit
$\rho^w$, and eq.\eqref{eq:born2.1} implies that there is coincidence also
of the ReIm bit $\rho$, i.e.~$\rho^\rI=\rho^{\rII}=\rho$.  The
resulting joint probabilities are marginalized over $\rho$, as in
standard QM, where Born's rule leads to the addition of the 
squares of real and imaginary parts of the wave
function. Alternatively, there could be just a single $\rho$, but both
twin worlds would have to interact then with the same corresponding
bit. \\

In any case, resorting to discrete values $\sigma,\rho$ leads to a $\bm X^w=(\bm
x^w,\sigma^w,\rho^w)\in\mathbb{R}^3\times \mathbb Z_2^2$. With
$\sigma^\rI,\sigma^\rII,\rho^\rI,\rho^\rII$ taken together, a particle that appears as
spinless particle  in Our World is reminiscent of 
four spin-1/2 particles in
the Full World. ``Reminiscent'' in the sense that in each
Twin World we have a ``spinor probability distribution'', i.e.~a
four-component probability vector where the difference between two
components with fixed $\rho$ determines one half of the realified wave
function (the other half given by the complementary value of $\rho$),
rather than those two components determining a 
two-component complex spinor wavefunction as usual.\\

One might be inclined to think that with the grabit approach and
the Twin World interpretation of quantum mechanics, we get a
hidden variable theory, with the extra variables compared to $\bm x$
in Our World playing the role of the hidden variables, and a global
probability distribution over all variables in the Full World.  One
would conclude that certainly no Bell inequality  
could then be violated.  This is, however, not true. 
In the next chapter I investigate 
We will see that 
the Twin World formalism allows one to exactly reproduce QM correlations, including
the violation of the Bell-type CHSH inequalities.

\section{Simple model systems}

\subsection{A single qubit with a phase shift}
As simplest example consider the single-qubit state
$\ket{\Psi}=(\ket{0}+e^{i\varphi}\ket{1})/\sqrt{2}$.  It can be created
by acting on the qubit in state $\ket{0}$ with a Hadamard gate and
then with the general phase gate
\begin{equation}
  \label{eq:Uphi}
  U(\varphi)=\begin{pmatrix}
    1 & 0\\
    0& e^{i\varphi}
    \end{pmatrix}\,,
  \end{equation}
see FIG.~\ref{fig:qcSingQubit}.
  $U(\varphi)$ maps to a corresponding realified gate
  $\tilde{U}(\varphi)$, representable in the basis of realified states as
  \begin{equation}
    \label{eq:Utilde}
    \tilde{U}(\varphi)=\begin{pmatrix}
    \bm 1_2 & \bm 0_2\\
    \bm 0_2& \begin{matrix}
    \cos \varphi & -\sin \varphi\\
    \sin\varphi& \cos\varphi
    \end{matrix}
  \end{pmatrix}\,,
  \end{equation}
where $(i,\rho)$ ordering of indices was used.  When acting on
$\ket{\Phi}=(1,0,1,0)^t/\sqrt{2}$, the realified version of the state $(\ket{0}+\ket{1})/\sqrt{2}$ obtained after the Hadamard gate acting on the initial $\ket{0}$, $\tilde{U}(\varphi)$
creates the realified state 
$\ket{\Phi}=(\ket{00}+\cos\varphi\ket{10}+\sin\varphi\ket{11})/\sqrt{2}$,
representable by two real qubits.  
Hence, written as a real gate acting on a real state vector, the general phase
gate acts as a controlled real phase gate and generates a rotation
in the state space of the target-qubit (=ReIm qubit) if the first
(control-)qubit is in the state $\ket{1}$, see FIG.~\ref{fig:qcSingQubitReal}.\\

Recently, there has been substantial work targeted at showing that QM
must be described fundamentally by a complex wave function
\cite{renou_quantum_2021}.  Necessary for these proofs is the assumption that
only local operations are permitted.  Indeed, additional non-locality
arises in the
present example in the control-target structure.  Conversely, if one
insists on a purely local description of quantum mechanics and reality
of the wave function, one can find bounds on correlations, similar to
Bell inequalities.  These have been shown to be violated
experimentally, a result that was interpreted in the sense that a purely real
description of quantum mechanics is not possible \cite{renou_quantum_2021}. In
the present example, we see that the additional non-locality is rather
benign.  It results from the interaction between the original qubit and the
ReIm qubit. If $\rho$ is an internal spin-like degree of freedom, this
interaction is local in position space. 
In the case of a many-particle system, many different combinations of
the local $\rho_\nu^w$ contribute to the global $\rho^w$ in each Twin
World.  
This is the same kind of non-locality as in
standard quantum mechanics: a local phase can globally and instantaneously change a relative phase
of a component of an entangled state. 
\\

The translation of $\tilde{U}(\varphi)$ to a grabit gate was given in
\cite{braun_stochastic_2022}, and broken into different, permutationally related forms for
the four quadrants of the phase $\varphi$. In the subspace where the
control $c=1$, we get the stochastic map
\begin{equation}
  \label{eq:SRphi1}
  S_{R(\varphi)}^{(1)}=\frac{1}{1+q}\begin{pmatrix}
    q&0&0&1\\
    0&q&1&0\\
    1&0&q&0\\
    0&1&0&q
    \end{pmatrix}\,,
  \end{equation}
 in the quadrant $0<\varphi\le \pi/2$ 
  where $q=|\cot(\varphi)|$.   $S_{R(\varphi)}^{(1)}$ acts on the probability vector $(P_0,P_1,P_2,P_3)\in
  \mathbb{R}_+^4$, $\sum_{I=0}^3 P_I=1$, where $I$ enumerates the
  b4vs of the ReIm grabit.
In an emulation, the  stochastic process is implemented by acting on the drawn
  realizations and is defined  even for a
  single realization: E.g.~if the b4v is 0, we flip it with probability
  $1/(1+q)$ to 2 and keep it with probability $q/(1+q)$, and
  accordingly for the other b4v's. 
In the other quadrants, the signs of $\cos\varphi$ or $\sin\varphi$
change, and we get three more $S_{R(\varphi)}^{(i)}$.  A sign change of $\cos\varphi$ is implemented by moving $q$ in
$S_{R(\varphi)}^{(1)}$ one row up or down while keeping the same blv, see
\cite{braun_stochastic_2022} for details. 
As noted there, the pure rotation has to be combined with an amplitude reduction gate that reduces the 2-norm of the state in the $c=0$ subspace by the same amount as in the $c=1$ subspace, where $c$ stands for the control qubit.  The amplitude reduction gate can be expressed as  
\begin{equation}
  \label{eq:A}
  R_2=\begin{pmatrix}
    1-r_0& r_0 \\
    r_0 & 1-r_0\\
    \end{pmatrix}\,.
\end{equation}
with $r_0=(1-\mathcal{N})/2$ and $\mathcal{N}=\sqrt{1+q^2}/(1+|q|)$. It is implemented as a
stochastic process $0\leftrightarrow 1$ with probability $r_0$,
whereas with probability $1-r_0$ the two b4v's are not flipped, and
correspondingly for $2\leftrightarrow 3$. 
With this, ${U}(\varphi)$ on a single qubit is emulated with a 16
  $\times 16$ stochastic map on two grabits of the form
  \begin{equation}
    \label{eq:SfullRphi}
    S_{R(\varphi)}^{\text{full},i}=R_2\otimes \bm 1_4\oplus \bm 1_2\otimes S_{R(\varphi)}^{(i)}\,.
  \end{equation}
This leads to a translation of the quantum circuit of
FIG.~\ref{fig:qcSingQubitReal} to the grabit quantum circuit in each of
the two Twin Worlds shown in FIG.~\ref{fig:qcSingGrabit}. The
refreshment $\mathcal{R}$ at the end is necessary, as the two Hadamard
gates in series generate interference. If based on the probability
distribution or a histogram with infinitely many samples, it ensures that
$\tilde{p}_{i,\rho}^{I}=\tilde{p}_{i,\rho}^{II}=|\Phi_{i,\rho}|/||\Phi||_1$, where $\Phi$
is the true, realified, quantum mechanical state. Hence, in
each Twin World, the quantum mechanical state also becomes directly,
by the absolute values of its elements, a probability
distribution. Coincidence of identical outcomes from the two Twin
Worlds  marginalized over $\rho$ then gives exactly the statistics known from the Born-2 rule.

In FIG.~\ref{fig:qcSingResult} I show that this works: Up to small
statistical fluctuations due to the finiteness of the ensemble (here
$\Nl=100\,000$), the grabit quantum circuit from
FIG.~\ref{fig:qcSingGrabit} reproduces exactly the statistics of
standard quantum mechanics.

\begin{figure}[h]
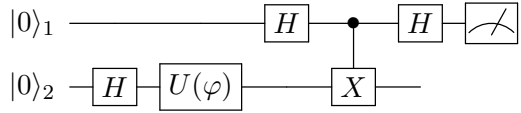

	\centering
	\begin{equation*}
		\qc{
                  \lstick{\ket{0}_1} &\qw     &\qw           &\gate{H}&\ctrl{1}& \gate{H} &\meter \\
                  \lstick{\ket{0}_2} &\gate{H}&\gate{U(\varphi)}&\qw     &\gate{X}& \qw
		}
	\end{equation*}
	\caption{Quantum circuit for the creation of a single-qubit quantum superposition $\ket{\psi(\varphi)}=(\ket{0}+e^{i\varphi}\ket{1})/\sqrt{2}$ in the second qubit and the subsequent measurement of $X$ with the help of the first  (ancilla) qubit measured in the computational basis.  
        } 
	\label{fig:qcSingQubit}
\end{figure}

\begin{figure}[h]
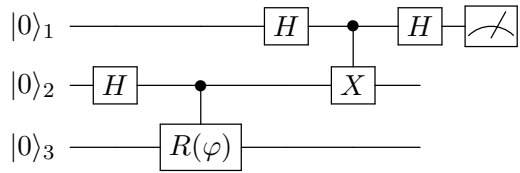

	\centering
	\begin{equation*}
		\qc{
                  \lstick{\ket{0}_1} &\qw     &\qw           &\gate{H}&\ctrl{1}& \gate{H} &\meter \\
                  \lstick{\ket{0}_2} &\gate{H}&\ctrl{1}&\qw     &\gate{X}& \qw\\
                  \lstick{\ket{0}_3} &\qw &\gate{R(\varphi)}&\qw     &\qw& \qw
		}
	\end{equation*}
	\caption{As FIG.~\ref{fig:qcSingQubit} but for the realified quantum state $\ket{\Phi}$. The third qubit works as ReIm-qubit, i.e.~codes the real or imaginary part of the state.  The complex $U(\varphi)$ gate on the second qubit becomes a controlled rotation gate $R(\varphi)$ on the ReIm qubit. 
        }
	\label{fig:qcSingQubitReal}
\end{figure}

\begin{figure}
	\centering
	\begin{equation*}
		\qc{
\lstick{(I)\mbox{ }\kett{0}} &\qw       &\qw           &\gate{S_H}&\ctrl{1}  & \gate{S_H} &\multigate{2}{\mathcal{R}} &\qw&\sgate{\mathcal{C}}{3}&\\
\lstick{\kett{0}} &\gate{S_H}&\ctrl{1}      &\qw       &\gate{S_X}& \qw        &\ghost{\mathcal{R}}        &&&\meter\\
\lstick{\kett{0}} &\qw       &\gate{S_{R(\varphi)}}&\qw    &\qw       & \qw        & \ghost{\mathcal{R}}       &&&\\
\lstick{(II)\mbox{ }\kett{0}} &\qw       &\qw           &\gate{S_H}&\ctrl{1}  & \gate{S_H} &\multigate{2}{\mathcal{R}} &\qw&\gate{\mathcal{C}}&\\
\lstick{\kett{0}} &\gate{S_H}&\ctrl{1}      &\qw       &\gate{S_X}& \qw        &\ghost{\mathcal{R}}        &&&\\
\lstick{\kett{0}} &\qw       &\gate{S_{R(\varphi)}}&\qw    &\qw       & \qw        & \ghost{\mathcal{R}}       &&&
		}
	\end{equation*}
	\caption{Two identical grabit quantum circuits in the two Twin
          Worlds ($\rI,\rII$) corresponding to
          FIG.~\ref{fig:qcSingQubitReal}. In each Twin World, the
          realified quantum state $\ket{\Phi}$ is encoded in three gradient
          bits (``grabits''), each realized with two classical
          stochastic bits, initialized in grabit state $\kett{0}$ (the double kets denote grabit states).
          Unitary gates become stochastic maps acting on those
          classical bits, and the complex single-qubit $U(\varphi)$
          becomes a controlled real rotation gate on the ReIm grabit.
          The refresh gates $\mathcal{R}$ make the final state
          interference free and assures that the physical probability
          distributions $\tilde{p}^\rI_{i}$ and $\tilde{p}^\rII_{i}$ over
          the outcomes $i=0,\ldots,7$ satisfy
          $\tilde{p}^\rI_{i}=\tilde{p}^\rII_{i}=|\Phi_i|/||\Phi||_1$,
          where $\Phi$ is the true, realified, quantum mechanical
          state. Post-selected coincident outcomes of the blvs from the two Twin Worlds
          marginalized over the second and third grabit define the observed
          outcomes and their statistics in Our World, and satisfy the
                    standard Born-2 rule,
          $\ptil_{i_1}=\sum_{i_2,i_3=0,1}{\ptil}_{i_1i_2i_3}=\sum_{i_2,i_3=0,1}\tilde{p}^\rI_{i_1i_2i_3}\cdot\tilde{p}^\rII_{i_1i_2i_3}=\sum_{i_2,i_3=0,1}\Phi_{i_1i_2i_3}^2=\sum_{i_2}|\Psi_{i_1i_2}|^2$,
          where $i_3$ labels the states of the ReIm grabit. 
                  } 
	\label{fig:qcSingGrabit}
\end{figure}

\begin{figure}[h]
    \centering
   \includegraphics[width=0.5\textwidth]{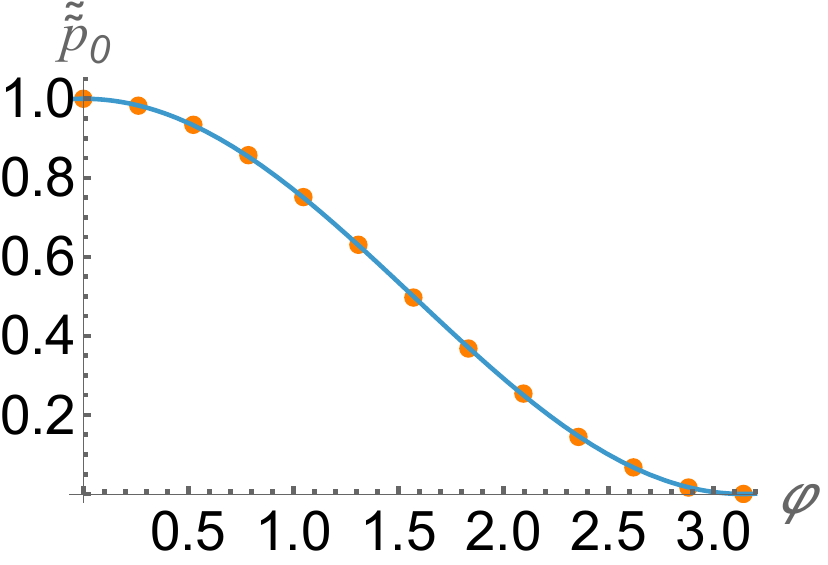}
    \caption{Rotation of the Bloch vector of a qubit in the $xy$-plane
      (probability of finding 0 in the computational basis measurement of the first qubit, $\ptil_0(\varphi)=(1+\langle X(\varphi)\rangle)/2=(1+\cos\varphi)/2$ in the state
      $\ket{\Psi(\varphi)}=(\ket{0}+e^{i\varphi}\ket{1})/\sqrt{2}$)
      calculated from the stochastic emulation with grabits in the two
      Twin Worlds and coincidence detection of the measurement
      outcomes 0 and 1. Blue continuous line: quantum mechanical
      prediction. Orange dots: results from the stochastic emulation (see
      FIG.~\ref{fig:qcSingGrabit}), $\Nl=100\,000$ }
    \label{fig:qcSingResult}
\end{figure}

\subsection{Two qubits and the violation of the CHSH inequality}
While all statistics of a single qubit can be reproduced with a
classical hidden variable model (see \cite{Bell87}, chapter 1), one can rule out local hidden
variable theories as modeled by Bell as valid descriptions of quantum reality by
examining the statistics of two qubits.  Let 
Alice and Bob each control one out of two qubits which are initially given to them in the singlet state $\ket{\Psi}=(\ket{01}-\ket{10})/\sqrt{2}$.  They obtain arbitrarily many copies, and then both measure randomly one of two observables, each with two possible outcomes  $\pm 1$: Alice measures $Q$ or $R$, Bob measures $S$ or $T$.  It is easy to show \cite{ClauserHSH69} that if Nature is described by a local hidden variable theory in the sense of Bell, the expectation values of their joint measurements must satisfy 
\begin{equation}
  \label{eq:CHSH}
  |\braket{QS}+\braket{RS}+\braket{RT}-\braket{QT}|\le 2
\end{equation}
Yet, quantum mechanics predicts, and numerous experiments have
verified (see \cite{hensen_loophole-free_2015} and references
therein), that by choosing $Q,R,S,T$ parameterized as  
\begin{equation}
  \label{eq:Qth}
  Q(\theta)=(\cos\theta) X+(\sin \theta) Z =\begin{pmatrix}
    \sin\theta& \cos\theta \\
    \cos\theta& -\sin\theta\\
    \end{pmatrix},\,\,\theta\in\mathbb{R},
\end{equation}
(which is at the same time hermitian and unitary with eigenvalues $\pm
1$), with values of $\theta_1,\theta_2$ as in table \ref{tab1} there
are values of $\varphi$ for which inequality \eqref{eq:CHSH} is violated.
The maximum violation is found for $\varphi=-\pi/4$  and gives
$2\sqrt{2}$ on the left hand side.   This proves that Nature cannot be
described by a local hidden variable theory as modeled by Bell, and the same is true by
the same argument, of course, for the Twin World construction of quantum reality if its statistical predictions agree with those of standard QM, which is what I am  going to show now. \\   

All of $Q(\theta)$ are represented as real matrices in the
computational basis, and also the initial singlet state is real in the
computational basis. There is, therefore, no need for a ReIm qubit or grabit in the present context
and we will skip it in the following.  The corresponding stochastic
matrix that represents the action of $Q(\theta)$ on a grabit  is, in
the first quadrant $0\le \theta < \pi/2$, given by 
\begin{equation}
  \label{eq:SQth}
  S_{Q(\theta)}^{(1)}=\frac{1}{1+q}\begin{pmatrix}
    1&0&q&0\\
    0&1&0&q\\
    q&0&0&1\\
    0&q&1&0
    \end{pmatrix}\,,
  \end{equation}
where now $q=|\cot \theta|$. The stochastic maps in the other three
quadrants of $\theta$ are again obtained by shifting the 1's or $q$'s
a line up or down when $\sin\theta$ or $\cos\theta$ changes sign,
respectively, yielding   
\begin{equation}
  \label{eq:SQth2-4}
  S_{Q(\theta)}^{(2)}=\frac{1}{1+q}\begin{pmatrix}
    1&0&0&q\\
    0&1&q&0\\
    0&q&0&1\\
    q&0&1&0
  \end{pmatrix}\,,
    S_{Q(\theta)}^{(3)}=\frac{1}{1+q}\begin{pmatrix}
    0&1&0&q\\
    1&0&q&0\\
    0&q&1&0\\
    q&0&0&1
  \end{pmatrix}\,,
    S_{Q(\theta)}^{(4)}=\frac{1}{1+q}\begin{pmatrix}
    0&1&q&0\\
    1&0&0&q\\
    q&0&1&0\\
    0&q&0&1
    \end{pmatrix}
  \end{equation}
  for $\pi/2\le \theta<\pi$, $-\pi\le \theta<-\pi/2$, and $-\pi/2\le \theta<0$, respectively.
  A measurement of $Q(\theta_1)$ and $Q(\theta_2)$ on qubits 2,4 can once again
  be accomplished with the aid of an ancilla qubit each and raising
  $Q(\theta_1)$ and $Q(\theta_2)$ to controlled gates \cite{nielsen_quantum_2011}. The control is by those ancilla   qubits, brought into a superposition of $\ket{0}$ and $\ket{1}$, see
  FIG.~\ref{fig:CHSHqubit} for the qubit case and FIG.\ref{fig:CHSHgrabit} for the grabit version. 
\begin{table}
  \centering
\begin{tabular}{c|c||c|c}
  $Q(\theta_1)$&$Q(\theta_2)$&$\theta_1$ & $\theta_2$\\\hline
  $Q$&$S$& $\pi/2$ & $\pi/2+\varphi$\\
  $R$&$S$& 0 & $\pi/2+\varphi$\\
  $R$&$T$& 0 & $\varphi$\\
    $Q$&$T$& $\pi/2$ & $\varphi$
  \end{tabular}\label{tab1}
  \caption{Measurement operators of Alice and Bob given by angles
    $\theta_1$, $\theta_2$ in \eqref{eq:Qth}. 
  } 
\end{table}
\begin{figure}
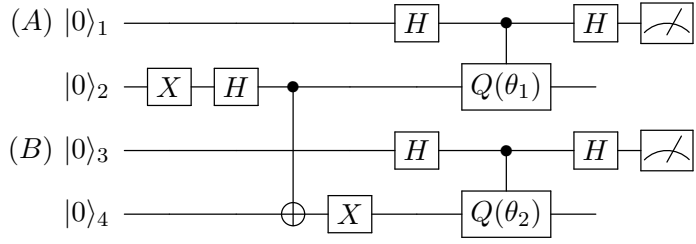

	\centering
	\begin{equation*}
		\qc{
\lstick{(A)\mbox{ }\ket{0}_1} &       \qw&\qw&\qw &\qw&\gate{H}&\ctrl{1}  & \gate{H}  &\meter\\
\lstick{\ket{0}_2} & \gate{X}&\gate{H}&\ctrl{2}&\qw     &\qw&\gate{{Q(\theta_1)}}          & \qw               &\\
\lstick{(B)\mbox{ }\ket{0}_3} & \qw&\qw&\qw&\qw             &\gate{H}&\ctrl{1}  & \gate{H}  &\meter\\
\lstick{\ket{0}_4} &  \qw&\qw&\targ&\gate{X}& \qw&\gate{{Q(\theta_2)}}         & \qw        &   		}
	\end{equation*}
	\caption{Quantum circuit for the realization of the CHSH
          experiment. The first four gates prepare a singlet state
          between qubits 2,4, on which $Q(\theta_1)$ and
          $Q(\theta_2)$ are then measured via a
          measurement of the ancilla qubits 1,3 in the computational
          basis, respectively.  Four different combinations of $\theta_1,\theta_2$
          are randomly chosen in a large number of runs, and empirical
          averages computed that estimate the expectation values in
          \eqref{eq:CHSH}.   The grabit emulation of the quantum
          circuit is shown in FIG.\ref{fig:CHSHgrabit} and the results are in FIG.\ref{fig.chshResult}.
        } 
	\label{fig:CHSHqubit}
\end{figure}

\begin{figure}
	\centering
	\begin{equation*}
		\qc{
\lstick{(A^{\rI})\mbox{ }\kett{0}} &       \qw&\qw&\qw &\qw&\gate{S_H}&\ctrl{1}  & \gate{S_H} &\multigate{3}{\mathcal{R}} &\meter\\
\lstick{\kett{0}} & \gate{S_X}&\gate{S_H}&\ctrl{2}&\qw     &\qw&\gate{S_{Q(\theta_1)}}          & \qw        & \ghost{\mathcal{R}}       &\\
\lstick{(B^{\rI})\mbox{ }\kett{0}} & \qw&\qw&\qw&\qw             &\gate{S_H}&\ctrl{1}  & \gate{S_H} &\ghost{\mathcal{R}} &\meter\\
\lstick{\kett{0}} &  \qw&\qw&\targ&\gate{S_X}& \qw&\gate{S_{Q(\theta_2)}}         & \qw        & \ghost{\mathcal{R}}&       
		}
	\end{equation*}
	\caption{Grabit quantum circuit for the realization of the
          CHSH experiment in one of the two Twin Worlds. With the
          chosen initial singlet state between grabits 2,4 and the
          four possible measurements from Table \ref{tab1}, the
          quantum state remains always real and a ReIm grabit is not
          needed and not shown. The first four gates prepare a singlet
          state between grabits 2,4, on which $Q(\theta_1)$ and
          $Q(\theta_2)$, respectively, are then measured via a
          measurement of the ancilla qubits 1,3 in the computational
          basis. All grabit gates are implemented as stochastic
          maps. $\mathcal{R}$ is the refreshment step that leads to
          $\tilde{p}^\rI_{i}=|\Phi_i|$ (and correspondingly in the
          second Twin World, not shown here,
          $\tilde{p}^\rII_i=|\Phi_i|$). In Our World, only the
          coincidences from the two Twin Worlds of all four measurement outcomes 00,01,10, and 
          11 are observed and lead to outcome statistics identical to
          the one predicted by standard quantum mechanics --- and hence a violation of the CHSH inequality, see FIG.~\ref{fig.chshResult}.
        } 
	\label{fig:CHSHgrabit}
\end{figure}

\begin{figure}[h]
    \centering
    \includegraphics[width=0.5\textwidth]{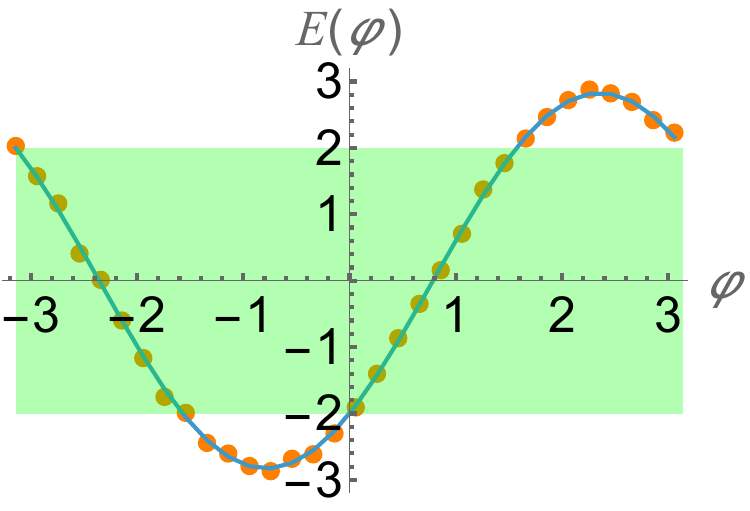}
    \caption{Violation of the CHSH inequality from the stochastic
      emulation with grabits in the two Twin Worlds and coincidence
      detection of the four measurement outcomes 00,01,10, and
      11. Blue continuous line: quantum mechanical prediction for $
      E(\varphi)\defas
      \braket{QS}+\braket{RS}+\braket{RT}-\braket{QT}$ with $Q,R,S,T$
      as defined in Table \ref{tab1}.   Orange dots:
      results from the stochastic emulation with $\Nl=10^4$. Any value absolutely
      larger than 2 signals a violation of the CHSH inequality and
      hence invalidates a description with a local realistic theory as modeled by Bell, whereas the values in the green rectangular area are compatible with that model. 
    }
    \label{fig.chshResult}
\end{figure}
In FIG.\ref{fig.chshResult} we see that the grabit emulation of the
CHSH experiment, including refreshments and post-selection from the
two Twin Worlds, leads to the same statistics as standard QM, and in
particular to a violation of the CHSH inequality for the same values
of $\varphi$ as in standard QM.

\section{Dynamics}\label{sec.dynami}
One might think that with the possibility of decomposing any unitary
matrix in an arbitrarily large Hilbert space of dimension $2^n$ into
elementary gates that can be emulated with stochastic matrices
\cite{braun_stochastic_2022}, one can also reproduce any quantum time evolution
with arbitrary hamiltonians.  However, the mentioned Hilbert space is
the one of $n$ qubits, whereas it is not so clear whether one can
emulate the quantum time evolution of a single particle in a large
Hilbert space in a natural way. In other words, one would like to find
a stochastic process that emulates Schr\"odinger's equation for a
non-relativistic quantum particle.  This is what we are going to
investigate in this Section.\\

\subsection{Single particle dynamics}\label{sec.dynami2}
Starting point is the unitary time evolution over an infinitesimal
time-interval $dt$,
\begin{equation}
  \label{eq:U1}
  U=\mathbb{1}-i H dt/\hbar+\ord(dt^2)\,,
\end{equation}
where the Hamiltonian for a single particle of mass $m$ in position representation reads
\begin{equation}
  \label{eq:H}
  H=-\frac{\hbar^2}{2m}\Delta +V(\bm x)\,,
\end{equation}
with $V(\bm x)$ the potential energy.
The unitary operator in eq.\eqref{eq:U1} propagates the true quantum
state $\ket{\Psi}$ according to
\begin{equation}
  \label{eq:psidt}
  \ket{\Psi(dt)}=U\ket{\Psi(0)}=\ket{\Psi(0)}-i H
  \frac{dt}{\hbar}\ket{\Psi{(0)}}+\ord(dt^2)\,.
  \end{equation}
We ``realify'' the state, $\Phi(\bm x)=(\Re \Psi(\bm x),\Im
\Psi(\bm x))$, with the wavefunction $\Psi(\bm x)=\braket{\bm
  x|\Psi}$. The chosen block-structure now implies that the ReIm bit
is taken as most significant bit, and so the index-ordering convention
in this Section will be $(\rho,\sigma, \bm x)$.  I will refer to the
triple $(\rho,\sigma, \bm x)$ as ``phased position variable'', ``ppv''
for short.
$\Phi(\bm x)$ is propagated by $\tilde{U}$,
\begin{eqnarray}
  \label{eq:Phidt}
  \Phi(\bm x, dt)&=&\tilde{U}\Phi(\bm x, 0)\\
  \tilde{U}&=&
\begin{pmatrix}
\mathbb{1}_N
  & \rvline & (-\Delta+\tilde{V}(\bm x)) d\tilde{t} \\
\hline
  (\Delta-\tilde{V}(\bm x)) d\tilde{t} & \rvline &
  \mathbb{1}_N
\end{pmatrix}\,,               
\end{eqnarray}
with $d\tilde{t}=(\hbar/(2m))dt$, $\tilde{V}(\bm x)=V(\bm
x)/(\hbar^2/(2m))$. 
Let us first consider a particle in one dimension with discretized
position space, i.e.~on a regular 1D lattice with $N$ equally spaced
sites with positions $x_i=i\, a$ and periodic boundary conditions.
The lattice constant $a$ is set to 1 in the following.  The
Laplace operator is then represented in matrix form as 
\begin{eqnarray}
  \label{eq:DelO1}
  \Delta&=&-2\,\mathbb{1}_N+O_1
\end{eqnarray}
where the matrix elements of $O_1$ are given by $(O_1)_{i,j}=\delta_{i,j+1}+\delta_{i,j-1}$, $i,j\in \{1,\ldots,N\}$, and indices $1,N+1$ are identified. This gives 
\begin{eqnarray}
  \label{eq:UtilDis}
  \tilde{U}&=&
\begin{pmatrix}
\mathbb{1}_N 
  & \rvline & (2\,\mathbb{1}_N-O_1+\tilde{V}) d\tilde{t} \\
\hline
  (-2\,\mathbb{1}_N+O_1-\tilde{V}) d\tilde{t} & \rvline &
  \mathbb{1}
\end{pmatrix}               
\end{eqnarray}
with $\tilde{V}\defas\text{diag}(\tilde{V}_1,\ldots,\tilde{V}_N)$,
$\tilde{V}_i\defas\tilde{V}(x_i)$. We try to reproduce this evolution
of $\Phi$ via a stochastic map $S$ of a probability vector
$\{P_I\}_{I=0,\ldots 4N-1}$ with triple indices $I=\{\rho,\sigma,x\}$
that gives rise to an emulation $\phi$ of the true quantum mechanical
$\Phi$ via 
\begin{equation}
  \label{eq:phirx}
\phi_{\rho\,x}=\sum_{\sigma=0,1}(-1)^\sigma P_{\rho\,\sigma\,x}\,,  
\end{equation}
where $\rho\in\{0,1\}$ decodes the ReIm bit.  
The ``sign index'' $\sigma\in\{0,1\}$ determines
whether $P_{\rho\sigma\,x}$ contributes to the positive or negative
part in \eqref{eq:phirx}, and  $x=1,\ldots,N$ codes the position of
the particle on the 1D lattice. As in \cite{braun_stochastic_2022} we
can understand $\phi_{\rho\,x}$ as the discretized derivative of a
classical probability distribution in a new direction $\sigma$.  The
stochastic map propagates $P$ according to  
\begin{equation}
  \label{eq:Pprop}
  P(\tilde{t}+d\tilde{t})_{\rho\sigma\,x}=S_{\rho\sigma\,x,\rho'\sigma'\,x'}P_{\rho'\sigma'\,x'}(\tilde{t})\,,
\end{equation}
where Einstein summation convention over repeated indices is used. The
$\sigma$ index doubles the size of $S$ compared to $\tilde{U}$.

\subsection{Free particle - pure kinetic energy}
Let us first consider the case of $V(x)=0$, i.e.~a free particle with only kinetic energy. It is then straight-forward to write down a stochastic matrix that emulates the quantum evolution:
Up to normalization,
\begin{equation}
  \label{eq:Stil}
  \tilde{S}=\mathbb{1}_{4N}+d\ttil
\begin{pmatrix}
  \mathbb{0}_N  &\rvline &\mathbb{0}_N &\rvline  &2\,\mathbb{1}_N&\rvline  & O_1\\
  \hline
  \mathbb{0}_N  &\rvline &\mathbb{0}_N &\rvline  &O_1 &\rvline & 2\,\mathbb{1}_N\\
  \hline
  O_1& \rvline & 2\,\mathbb{1}_N& \rvline &\mathbb{0}_N& \rvline & \mathbb{0}_N \\
  \hline
  2\,\mathbb{1}_N& \rvline & O_1 & \rvline &\mathbb{0}_N& \rvline & \mathbb{0}_N
\end{pmatrix}               
\end{equation}
reproduces the correct amplitude ratios in the exact $\Phi$ to order
$d\ttil$. Here, $(\mathbb{0}_N)_{ij}=0\,\forall i,j=1,\ldots,N$. 
The structure of \eqref{eq:Stil} reflects that phase factors of
$\ket{\Psi(0)}$ must propate to $\ket{\Psi(d\ttil)}$.  Let
$e_{\rho_0\sigma_0x_0}$ be a localized probability distribution
(``computational basis state'') with components
$(e_{\rho_0\sigma_0x_0})_{\rho\sigma
  x}=\delta_{\rho\rho_0}\delta_{\sigma\sigma_0}\delta_{xx_0}$, mapped
to $e_{\rho_0\sigma_0x_0}'$ with components  
$(e_{\rho_0\sigma_0x_0}')_{\rho\sigma x}$. The combinations
$(\rho_0\sigma_0)=(00),(01),(10),(11)$ correspond to phase factors
$1,-1,i,-i$ of $\braket{x|\Psi_0}$. A change of the initial phase of
$\braket{x|\Psi_0}$ with such a phase factor implies the same change
of phase factor of the final state. E.g.~$\Psi\to i\Psi$ implies
$\Psi'\to i \Psi'$, which is equivalent to 
\begin{equation}
  \label{eq:PhiTrafo}
  (\Phi_{0\sigma x},\Phi_{1\sigma x})\to (-\Phi_{1\sigma x},\Phi_{0\sigma x})\,,
\end{equation}
for the realified state. Such transformations relate blocks of $P'_{\rho_0\sigma_0x_0}$ with different $\rho_0,\sigma_0$ to each other, and imply a block structure of $S$,
\begin{equation}
  \label{eq:Sblock}
  \tilde{S}=\mathbb{1}_{4N}+d\ttil
\begin{pmatrix}
A&\rvline &B&\rvline  &D &\rvline  & C\\
  \hline
B&\rvline &A&\rvline  &C &\rvline  & D\\
  \hline
C&\rvline &D&\rvline  &A &\rvline  & B\\
  \hline
D&\rvline &C&\rvline  &B &\rvline  & A\\
\end{pmatrix}
\end{equation}
that is indeed obeyed by \eqref{eq:Stil}. \\

$\Stil$ is not a stochastic matrix yet: we also need that all columns
be normalized to 1-norm 1. To order $d\ttil$ we have for a column
$\Stil_i$ before renormalization
\begin{equation}
  \label{eq:normStil}
  ||\Stil_i||_1=1+4 d\ttil\,.
\end{equation}
Dividing each column by that norm gives to order $\mathcal{O}(dt)$ 
\begin{eqnarray}
  \label{eq:Sfin}
S&=&\mathbb{1}_{4N}+d\ttil
\begin{pmatrix}
  -4\,\mathbb{1}_N  &\rvline &\mathbb{0}_N &\rvline  &2\,\mathbb{1}_N&\rvline  & O_1\\
  \hline
  \mathbb{0}_N  &\rvline &  -4\mathbb{1}_N &\rvline  &O_1&\rvline & 2\,\mathbb{1}_N\\
  \hline
  O_1& \rvline & 2\,\mathbb{1}_N& \rvline &  -4\mathbb{1}_N& \rvline & \mathbb{0}_N \\
  \hline
  2\,\mathbb{1}_N& \rvline & O_1 & \rvline &\mathbb{0}_N& \rvline &   -4\mathbb{1}_N
\end{pmatrix}\\
  &\asdef&\mathbb{1}_{4N}+d\ttil\,G_T\,,\label{eq:GT}
\end{eqnarray}
where the last line defines the generator $G_T$ of the stochastic map over
an infinitesimal time step for pure kinetic energy.\\

The entries in all columns of $S$ are identical up to
permutation.  As a consequence, the 2-norms of the quantum states
obtained from propagating a canonical basis state are independent of
$x$.  
Therefore, in a superposition of initial position eigenstates,
$\ket{\Psi}=\sum_x c_x\ket{e_x}$, the reduced 2-norms of the 
propagated position eigenstates $\ket{e'_x}$ do not change the relative weights of the $\ket{e'_x}$ given by the 
$|c_x|^2$, in agreement with standard QM. Hence, for a free particle there is no need for an
``amplitude reduction gate'' \cite{braun_stochastic_2022}.\\
 $G_T$ defined in \eqref{eq:GT} is the generator of
an infinitesimal stochastic map. However, if the full $4N$-dimensional probability
distribution evolved according to $\dot{P}(t)= G_T P(t)$, where the dot means
time-derivative,   this would
lead to an exponential reduction of the contrast of $\ket{\phi(t)}$ as
function of time $t$. 
In principle, for
reproducing Schr\"odinger's equation in each of the two Twin Worlds, this does not matter,
as we can renormalize $\ket{\phi(t)}$ to $||\ket{\phi(t)}||_2=1$ at all
times.  However, in order to have $|\braket{x|\phi}|=\tilde{p}_x$ in 
each of the two Twin Worlds, the generator $G$ must be followed by a nonlinear
``refreshment operation'' \cite{braun_stochastic_2022} $\mathcal{R}$.   
Besides slowing down the loss of contrast in numerical emulations, the
refreshment step has the important consequence of making the state 
``interference-free'', i.e.~{\em either} $P_{\rho x}^{\rI,\rII}\defas P_{\rho 0 x}+P_{\rho 1
  x}=P_{\rho 0 x}$ {\em or}  $P_{\rho x}^{\rI,\rII}=P_{\rho 1 x}$. In either case
$P_{\rho x}^{\rI,\rII}=|\phi_{\rho x}|=\tilde{p}_{\rho x}^{I,II}$. As this stochastic dynamics
happens independently and identically in the two Twin Worlds, we
have the desired Born-2 rule, $\dbtilde{p}_{\rho x}=|\phi_{\rho x}|^2/\sum_{x}|\phi_{\rho x}|^2$ in Our
World.   \\

Figure \ref{fig.1DSE} shows the result of a numerical simulation of
the 1D Schr\"odinger equation for a free particle on a regular lattice with periodic
boundary conditions.  The wave function is plotted for a given time as
function of position.  Deviations from exact 
agreement can be made arbitrarily small by reducing the time step
$d\ttil$ used in the numerics. The figure also shows that the standard
spreading of an initially localized wave package is reproduced
correctly.
\begin{figure}[h]
    \centering
  (a)  \includegraphics[width=0.4\textwidth]{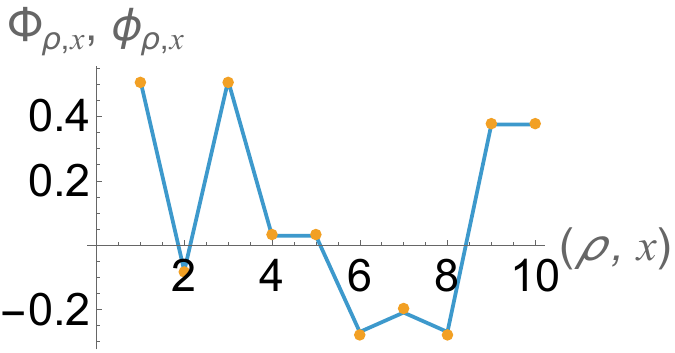}
  (b) \includegraphics[width=0.4\textwidth]{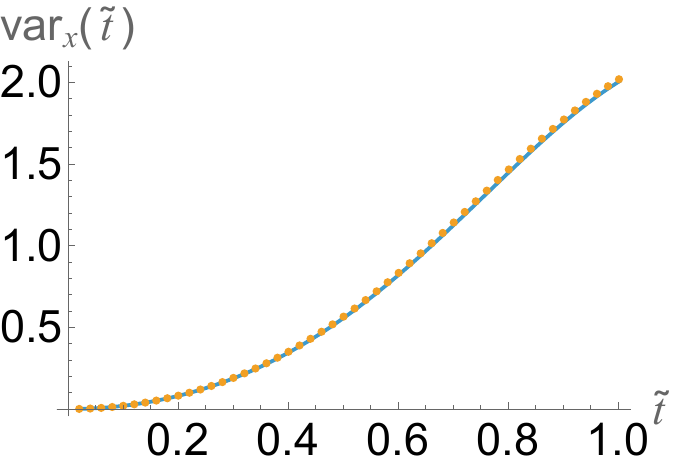}
    \caption{Simulation of the time evolution of the wave function of a
      free particle on a 1D regular lattice with $N=5$ by stepwise
      evaluation of the stochastic evolution
      eq.\eqref{eq:P''}. Real initial state localized on $x=2$.  (a) Position-dependent realified wave function
      $\phi_{\rho x}$ at $\tilde{t}=1$ (orange dots) from the
      stochastic emulation with 500 time-steps after normalization to 2-norm=1  compared
      to the 
      exact quantum mechanical realified wave function $\Phi_{\rho x}$ (blue continuous
      line, guide to the eye only --- the wave function is defined
      only on the lattice sites). (b) Variance of the
    wave package as function of time from the simulation
    (orange dots) compared to the exact quantum mechanical result (blue continuous line) 
    }
    \label{fig.1DSE}
\end{figure}

\subsection{Pure potential energy}
Quantum mechanical evolution over a time interval $t$ with only potential energy for a particle on the 1D line and $N$ discretized positions $x$, i.e.~a unitary $U=e^{-i V t/\hbar}$ leads to a realified orthonormal
\begin{equation}
  \label{eq:UtilV}
  \tilde{U}=\oplus_{x=1}^N\begin{pmatrix}
    \cos\varphi(x)& -\sin\varphi(x)\\
        \sin\varphi(x)& \cos\varphi(x)\\
    \end{pmatrix}
\end{equation}
in $(x,\rho)$ ordering, where $\varphi(x)=-\ttil\tilde{V}(x)$. Since
$V(x$) is diagonal in the position basis, we just get phase factors
for each position eigenstate, and a stochastic map emulating this
dynamics, including the amplitude reduction gates, can be 
can be found for arbitrary $\varphi(x)$.
Since $\tilde{V}$ can have positive and
negative parts, we decompose it with the help of the rectifier function,
$\text{rf}(x)\defas x \Theta(x)$, where $\Theta(x)$ is the Heaviside
function, according to 
\begin{eqnarray}
  \label{eq:Vtil}
  \tilde{V}&=&\rf(\tilde{V})-\rf(-\tilde{V})\\
    |\tilde{V}|&=&\rf(\tilde{V})+\rf(-\tilde{V})\,.
\end{eqnarray}
According to the recipe developed previously, it would be tempting  to distribute the positive entries of $\tilde{U}$ into entries of a stochastic matrix with indices $\sigma=0$, and the negative ones into those with indices $\sigma=1$.  However, a  non-constant potential $V(x)$ would then make that the 2-norms of  the columns of the emulated $\tilde{U}$ depend on $x$. If not corrected, this would lead to deviations
from standard quantum mechanics for the propagation of superpositions
of canonical basis states $\ket{e_{\rho_0x_0}}$,
\begin{equation}
  \label{eq:phi0}
\ket{\phi}=\sum_{\rho_0,x_0} c_{\rho_0\,x_0}\ket{e_{\rho_0\,x_0}}\,,  
\end{equation}
mapped via $S$ to $\ket{\phi'}=\sum_{\rho_0,x_0}
c_{\rho_0\,x_0}\ket{e_{\rho_0\,x_0}'}$.
One must therefore introduce a position-dependent amplitude reduction,
similar to the amplitude reduction gate developed in
\cite{braun_stochastic_2022} for the phase gate. We can do so by an
ansatz that once more exploits the gauge degree of freedom by writing a column with given $x$ of $S$ in
$(\rho,\sigma)$ ordering and skipping the zero entries as
$((1-r_x)\cos\varphi_x,r_x\cos\varphi_x,(1-r_x)\sin\varphi_x,r_x\sin\varphi_x)^t/(\sin\varphi_x+\cos\varphi_x)$
for $\pi/2>\varphi_x\ge 0$, and
$((1-r_x)\cos\varphi_x,r_x\cos\varphi_x,r_x\sin\varphi_x,(1-r_x)\sin\varphi_x)^t/(\sin\varphi_x+\cos\varphi_x)$
for $-\pi/2<\varphi_x< 0$, where $r_x$ with $1\ge r_x\ge 0$ is to be determined such that
the 2-norms of all columns are identical. 
One easily checks that if we make a column-vector interference-free,
i.e.~$r_x=0$, the 2-norm of the corresponding state $\phi_x'$ is given by
$1/\sqrt{1+\sin 2\varphi}$.  The smallest 2-norm 
arises hence from $\varphi$ closest to $\pi/4$.  Let us denote the corresponding $x$ as
$x_0$, i.e.~the
2-norms to which the 2-norms of all other $\phi'_x$ have to be
reduced by appropriate choice of the $r_x$ is $1/\sqrt{1+\sin 2\varphi_0}$.
For
finite $|\tilde{V}|$ and infinitesimal time $d\ttil$,  $\varphi(x)$ is
in $[0,\pi/2]$ for $V(x)<0$ and in $[-\pi/2,0]$ for $V(x)>0$. The smallest 2-norm 
arises then from  the largest absolute value $|V(x)|$.   

Altogether, with an
expansion to order $\ord(d\ttil)$, one finds that
the following stochastic matrix, written in $(\rho,\sigma, x)$
ordering, reproduces correctly the propagation of the realified
quantum state to the same order: 
\begin{eqnarray}
S_V&=&\mathbb{1}_{4N}+d\ttil
\begin{pmatrix}
-(\tV_0+\tV)/2&\rvline &(\tV_0-\tV)/2) &\rvline  &\tV_+&\rvline  & \tV_-\\
  \hline
  (\tV_0-\tV)/2) &\rvline &  -(\tV_0+\tV)/2&\rvline  &\tV_-&\rvline & \tV_+\\
  \hline
  \tV_-& \rvline & \tV_+& \rvline &  -(\tV_0+\tV)/2& \rvline & (\tV_0-\tV)/2)\\
  \hline
  \tV_+& \rvline & \tV_- & \rvline &(\tV_0-\tV)/2& \rvline &  -(\tV_0+\tV)/2
\end{pmatrix}\\
  &\asdef&\mathbb{1}_{4N}+G_V\,d\ttil\,,\label{eq:GV}
\end{eqnarray}
with $\tV_0\defas (\max_{x}|\tV_x|)\,\mathbb{1}_N$,
$\tV_+\defas \text{diag}(\text{rf}(\tV_x))$, $\tV_-\defas
\text{diag}(\text{rf}(-\tV_x))$, and $\tV=\text{diag}(|\tV_x|)$.

\subsection{Combination of kinetic and potential energy}\label{sec.comkipo}
To order $d\ttil$ the generators $G_T$ and $G_V$ simply add to give the full stochastic matrix,
\begin{equation}
  \label{eq:Sfull}
  S=\mathbb{1}_{4N}+d\ttil (G_T+ G_V)\,.
\end{equation}
One checks that \eqref{eq:Sfull} is indeed a bona-fide stochastic
matrix: all its entries are positive-semidefinite as long as $d\ttil$
is taken as infinitesimal and the scalar $\tV_0<\infty$, and the sum of all
entries in a column gives 1. To order $\ord(d\ttil)$, the
amplitude reduction implemented for the potential energy  also
equalizes the 2-norms of images of canonical basis states for the
combined map, with corrections of order $\ord(d\ttil^2)$.
Alternatively, one can define an amplitude reduction for both $G_T$
and $G_V$ combined that equalizes the
2-norms exactly, but since also the amplitude ratios resulting from
\eqref{eq:Sfull} are only correct to order $\ord(d\ttil)$, there
appears to be no need for this.  

Taken by itself, the stochastic dynamics generated by \eqref{eq:Sfull}
is not enough yet, however, to reproduce quantum mechanics: after each
infinitesimal time step, the obtained image $\phi'$ of the realified
quantum state $\phi$ will typically contain interferences that prevent
that the probabilities for the physical outcomes $x$ in each Twin
World are given by $\tilde{p}_x'=|\phi_x'|$. Hence, a refreshment
operation $\cal{R}$ must follow each infinitesimal time step that makes the
probability distribution $P'$ interference-free and enforces
$\tilde{p}_{x\rho}'=|\phi_{x\rho}'|$ for all $x,\rho$, leading to the chain of
equalities in  \eqref{eq:born2.5} to \eqref{eq:born2.1}. 
We may postulate that
in Nature refreshments are part of the propagation process. In the
present context, we can
define them formally as follows: 
\begin{definition}
Let $\cal{P}$ be the space of probability distributions of ppv's, i.e.~in the
present example of a single particle in 1D on $N$ sites the space of functions
$\mathbb{Z}_2^{2}\times \mathbb{Z}_N\to \mathbb{R}^+$.   A
refreshment is a map $\cal{R}:\cal{P}\to\cal{P}$, $P\mapsto {\cal
  R}(P)$ with
\begin{eqnarray}
  \label{eq:P''}
 ({\cal R(P)})_{\rho 0 x}&\defas& \left\{
  \begin{array}{cc}
    \phi_{\rho x} & \text{if }\phi_{\rho x}\ge 0\\
    0 & \text{if }\phi_{\rho x}< 0\\
  \end{array}
  \right.\nonumber \\
 ({\cal R(P)})_{\rho 1 x}&\defas& \left\{
  \begin{array}{cc}
    0 & \text{if }\phi_{\rho x}\ge 0\\
        -\phi_{\rho x} & \text{if }\phi_{\rho x}< 0
  \end{array}
\right.\,,
\end{eqnarray}
where $\phi_{\rho x}=\phi_{\rho x}(P)$ is defined in eq.\eqref{eq:phirx}.
\end{definition}
So the sign of $\phi_{\rho x}$ determines an index dynamics that we can also
write more compactly as 
\begin{eqnarray}
  \label{eq:P''2}
  ({\cal R}(P))_{\rho,\, (1-\sign{\phi_{\rho x}})/2,\, x}&=&
                                                             |{\phi_{\rho
                                                             x}}|
                                                             \nonumber
  \\
  ({\cal R}(P))_{\rho,\, (1+\sign{\phi_{\rho x}})/2,\, x}&=& 0\,,
\end{eqnarray}  
with commas inserted in the triple indices for better
distinguishability.

The full propagation over an infinitesimal
time-step $d\tilde{t}$ is then given by the fundamental equation
\begin{equation}
  \label{eq:propFull}
  P(\tilde{t}+d\tilde{t})={\cal R}(P(\tilde{t})+d\ttil\, G P(\tilde{t}))\,,
\end{equation}
with $G=G_T+G_V$. The form of the non-linearity (notably the fact that
$|\phi|$  is not differentiable at $\phi=0$),
makes that \eqref{eq:propFull} cannot be simply rewritten, by
expanding in $d\tilde{t}$, 
as a differential equation with a time-derivative on the left. 
Remarkably, while being non-linear due to the
refreshment, by construction and up to the different normalization, it
exactly reproduces the time evolution of the quantum mechanical
wave function $\Phi(x,t)$ propagated via the linear Schr\"odinger
equation, in the sense that the wave function $\phi(x,t)$ emulated via
eq.~\eqref{eq:phirx} from $P(x,t)$ obeys
$\phi(x,t)=\Phi(x,t)/||\Phi(x,t)||_1$ for all $x,t$. As explained
above, the different
normalization compared to QM cancels in the post-selected coincidence probabilties from the two
Twin Worlds such that the physical probabilities in Our
World are given exactly by the quantum mechanical ones,
$\dbtilde{p}(x,t)=|\Psi(x,t)|^2$. 

In FIG.\ref{fig:1Dtun} I show that \eqref{eq:propFull} correctly reproduces
another hall mark of quantum mechanics, namely the tunneling effect.
A wave package arriving from the left on a 1D line impinges on a
rectangular potential barrier. It is mostly reflected for the parameters
chosen, but a small part penetrates the barrier and arrives on the other
side. 
\begin{figure}[h]
    \centering
 (a)    \includegraphics[width=0.31\textwidth]{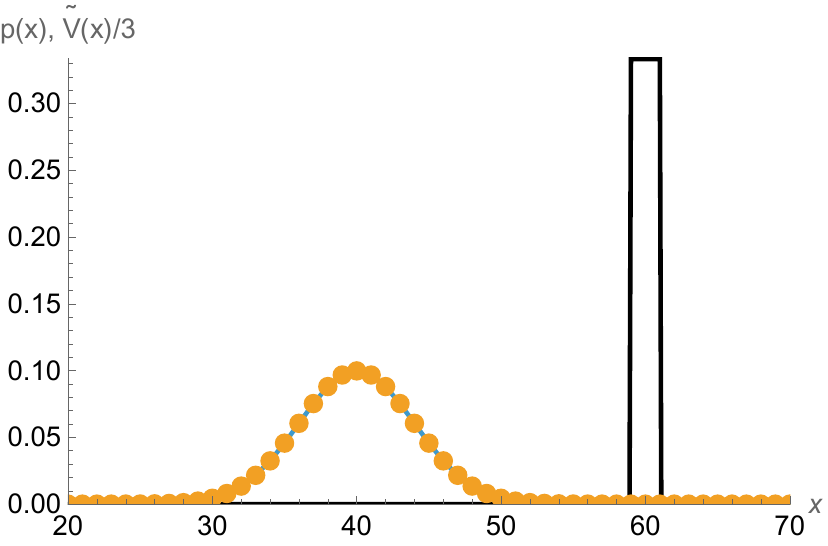}
     \includegraphics[width=0.31\textwidth]{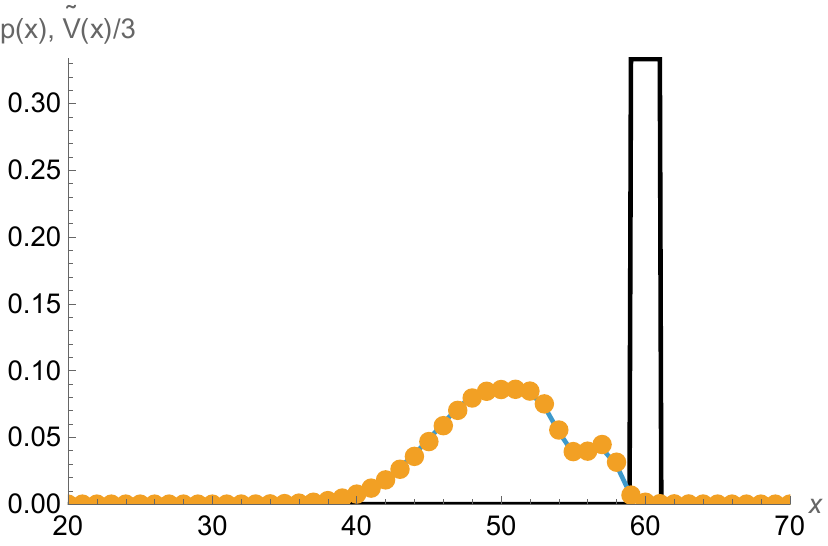}
     \includegraphics[width=0.31\textwidth]{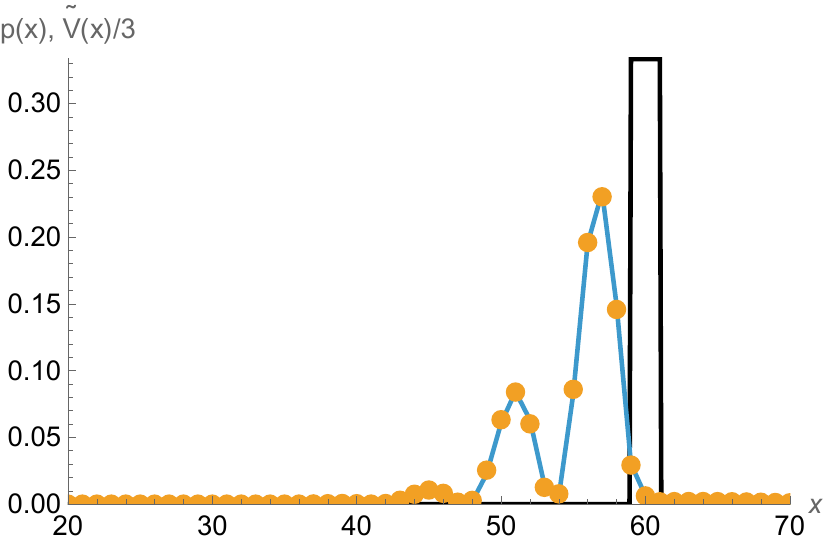}\\
     \includegraphics[width=0.31\textwidth]{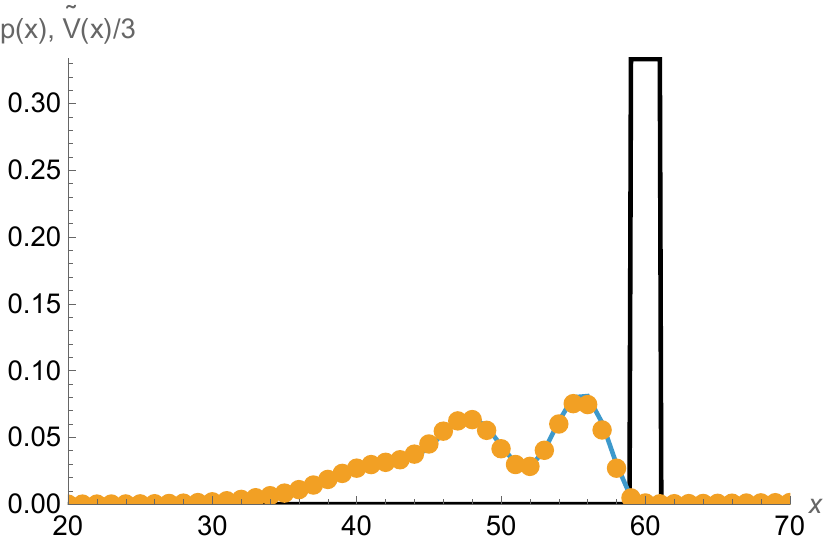}
     \includegraphics[width=0.31\textwidth]{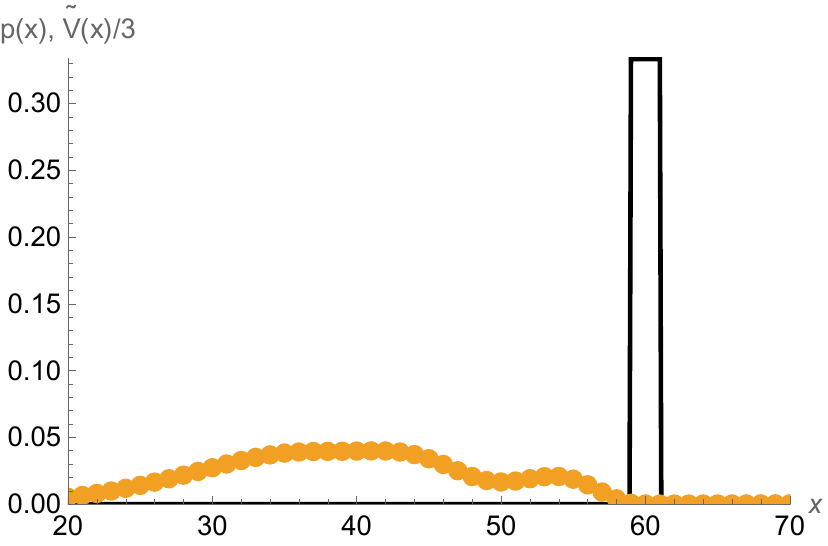}\\
     (b)     \includegraphics[width=0.31\textwidth]{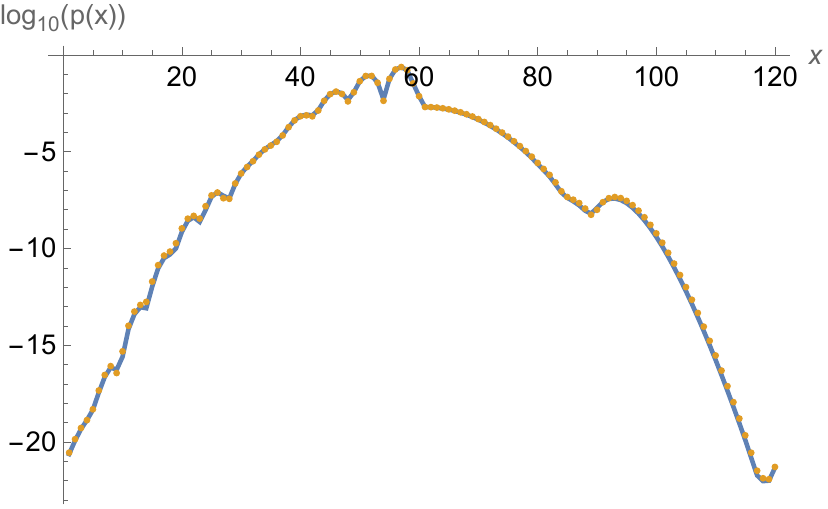}
     (c)     \includegraphics[width=0.31\textwidth]{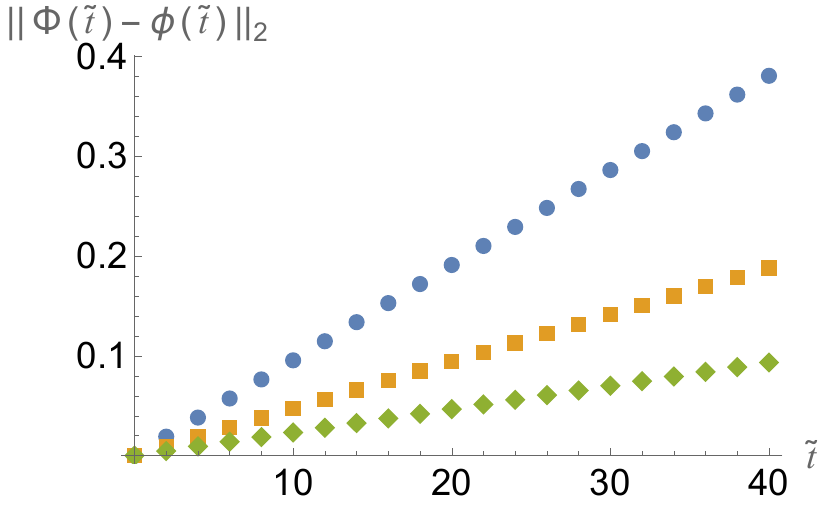}
     \caption{Twin World simulation of the quantum tunneling effect
       based on the stochastic evolution equation \eqref{eq:propFull}. The barrier is given by $\tV(x)=1$ for $59\le x \le 61$ and
      zero else (continuous black line).  For compatibility of scale, $\tV(x)/3$ is plotted.
      (a) Five snapshots of the time evolution of the
      probability density $p(x)$ at times
      $\tilde{t}=0$,       $\tilde{t}=10$,       $\tilde{t}=20$, $\tilde{t}=30$, and
      $\tilde{t}=\ttil_\text{max}=40$. Blue continuous lines represent the exact quantum mechanical
      result, dots the simulation using \eqref{eq:Sfull} in
      each Twin World with periodic boundary conditions, with a time step 
      $d\ttil=\ttil_\text{max}/(N_t-1)$ with $N_t=8001$.  The initial wave package is
      $\Psi(x,0)=\braket{x|\Psi(0)}=e^{-\frac{(x-x_0)^2}{4 \sigma_x^2}}e^{i 2 \pi  k
        x/N}/\cal{N}$ with normalization $\cal{N}$ such that
      $||\ket{\Psi(0)}||_2=1$ with $x_0=40$, $k=10$, $\sigma_x=4$, and
      $N=120$.   (b)
The two probability distributions at $\ttil=10$ and with $N_t=16001$
plotted logarithmically, $\log_{10} p(x)$, show numerical agreement of the Twin World emulation with standard QM on the level of $10^{-20}$.  
(c)      2-norm of the difference between 
      exact quantum mechanical evolution and the stochastic
      emulation, $||\Phi-\phi||_2$, as function of time $\ttil$ for different
      values of $N_t$: $N_t=4001$ (blue circles), $N_t=8001$ (red squares), $N_t=16001$ (green diamonds). The 2-norm increases linearly with time, but 
      the increase can be made arbitrarily small by choosing
      sufficiently small time steps.}  
    \label{fig:1Dtun}
\end{figure}
Besides allowing the chain of equalities \eqref{eq:born2.5} to
\eqref{eq:born2.1} that allows the Twin World interpretation of
quantum mechanics, the refreshments after each time step also have the
practical consequence of even allowing the stochastic emulation of the
quantum dynamics for larger times.  Without the refreshments, the
amplitudes of the realified images $\phi'$ decay exponentially 
with the number of interference-generating steps.  This leads to a
phase-transition like behavior, where all of a sudden the emulation
shows very large deviations from the true quantum state. This is
expected at the latest when
the amplitude has decayed to numerical machine precision.  \\

\subsection{Higher dimensions and multiparticle dynamics} 
It is straight forward to extend the previous results to $M\ge 1$ particles in
any other integer number $D$ of dimensions. Position eigenstates are 
then the usual joint position eigenstates, i.e.~tensor product states
$\ket{{\bm x}_1,\ldots,{\bm x}_M}$ with ${\bm x}_i\in\mathbb{R}^{D}$
  in the continuous case, or ${\bm x}_i\in\mathbb{Z}_{N^{D}}$ for a
  lattice of $N$ discrete positions in all $D$ directions, and
  $j=1,\ldots, M$.
The total hamiltonian is now
  \begin{equation}
    \label{eq:Htot}
    H=\sum_{j=1}^MT_j+\sum_{j=1}^MV_j(\bm x_j)+U({\bm x}_1,\ldots,{\bm x}_M),,
  \end{equation}
where $U$ contains all particle interactions, whereas $V_j$ is
an externally applied potential for particle $j$, as
before.  In the case of  two-body interactions, $U({\bm x}_1,\ldots,{\bm
  x}_M)=\sum_{j\ne k}U({\bm x}_j-{\bm x}_k)$, but as long as $U$ is
diagonal in the position eigenbasis, everything goes through for
arbitrary interactions, including multi-body interactions. 
For the kinetic energy, we simply have to replace the
discrete Laplace operator in 1D by the corresponding one in $D$
dimensions for $M$ particles,
$\Delta=\sum_{i=1}^M\Delta_i=\sum_{i=1}^M\sum_{j=1}^D{\partial^2}/{\partial
x_{i,j}^2}$, where $\bm x_i=(x_{i,1},\ldots,x_{i,D})$.  One checks that
because of the additivity of the Laplacian over all dimensions and particles
$G_T$ simply changes to 
\begin{eqnarray}
  \label{eq:GTDM}
G_T&=&
\begin{pmatrix}
  -4 D\,M\,\mathbb{1}_{N^{D\,M}}  &\rvline &\mathbb{0}_{N^{D\,M}} &\rvline  &2D\,M\,\mathbb{1}_{N^{D\,M}}&\rvline  & \sum_{i=1}^D\sum_{j=1}^MO_1^{(i,j)}\\
  \hline
  \mathbb{0}_{N^{D\,M}}  &\rvline &  -4D\,M\,\mathbb{1}_{N^{D\,M}} &\rvline  &\sum_{i=1}^D\sum_{j=1}^MO_1^{(i,j)}&\rvline & 2D\,M\,\mathbb{1}_{N^{D\,M}}\\
  \hline
  \sum_{i=1}^D\sum_{j=1}^MO_1^{(i,j)}& \rvline & 2D\,M\,\mathbb{1}_{N^{D\,M}}& \rvline &  -4D\,M\,\mathbb{1}_{N^{D\,M}}& \rvline & \mathbb{0}_{N^{D\,M}} \\
  \hline
  2D\,M\,\mathbb{1}_{N^{D\,M}}& \rvline & \sum_{i=1}^D\sum_{j=1}^MO_1^{(i,j)} & \rvline &\mathbb{0}_{N^{D\,M}}& \rvline &   -4D\,M\,\mathbb{1}_{N^{D\,M}}
\end{pmatrix}\,,
\end{eqnarray}
where 
\begin{eqnarray}
  \label{eq:O1ij}
  O_1^{(i,j)}\defas \mathbb{1}_{N^{D\,(j-1)+i-1}}\otimes O_1 \otimes \mathbb{1}_{N^{D(M-j)+D-i}}\,.
\end{eqnarray}
Again, no amplitude reduction is needed for pure kinetic energy.

For the potential energy, the interactions between the
particles can be
treated exactly as the external potential.  Indeed, it is convenient to
combine both in a total potential energy, $W=\sum_j V_j + U$.
Amplitude reduction works then as before, since QM just gives a diagonal
matrix of phase factors. In
$G_V$ (see eq.\eqref{eq:GV}) we can then simply replace $V\mapsto W$ everywhere, and hence
get a total generator $G=G_T+G_W$.  The maximization over $x$ in the
definition of $\tilde{V}_0$ gets replaced by a maximization over all
$({\bm x}_1,\ldots,{\bm x}_M)$, i.e.~$\tilde{W}_0\defas
(\text{max}_{{\bm x}_1,\ldots,{\bm x}_M}|W({\bm x}_1,\ldots,{\bm
  x}_M)|)\mathbb{1}_{N^{D\,M}}$.  With still keeping just two bits for
sign and ReIm, all stochastic matrices acting on the full
joint-probability distributions of all particle positions, and in
particular the identity operator in $S$ for $\ttil=0$,  have
now dimension $4N^{D M}$. 

\section{Discussion}\label{sec.discus}
\subsection{Particle number fluctuations}
The introduction of a Twin World and the additional dimension
associated to it begs the question why we should not see   
particles move in 
and out of Our World, and correspondingly have the number of particles
in Our World fluctuate? This concern motivated the introduction of
curled-up dimension in previous theories with extra dimensions, which,
however, is clearly not an option in the present theory as explained
in Section \ref{sec.TwiWor}.   
A first answer would be that for $N$ particles in Our World, the
fluctuations that we 
see are expected to be of order $\sqrt{N}$, i.e.~on relative order
$1/\sqrt{N}$. For a number of particles roughly estimated to $N\sim
10^{80}$ in the visible universe, they would  have gone entirely  
unnoticed. But presumably in a small subvolume with
much smaller number of particles, relative fluctuations of the number
of particles in that volume would be much larger. If we trapped a
single ion in a trap, say, would we not randomly see it being there or
not?  Experiments with neutral trapped atoms  achieve now the control
of the number of particles down to single particle-resolution in
ensembles up to more than thousand atoms
\cite{PhysRevLett.111.253001}.  In optical tweezer arrays, single
atoms trapped in dipole traps can be kept for thousands of seconds and
continuously imaged with fluorescence imaging with good theoretical
understanding of the remaining losses, and without needing explanation
of disapparance into higher dimensions
\cite{PhysRevApplied.16.034013}. A theory that predicts the
possibility of particles to disappear into higher dimensions would
hence be in contradiction to experimental evidence.

I now show that the Twin World picture of quantum reality
avoids this problem by the bias of reality in Our World made up of
post-selected coincidence events.  Consider for this purpose a
drastically simplified model of a measurement process that tries to
find out whether a particle is still somewhere in Our
World.  The key observation is that for a measurement process to take
place, one needs at least two particles: the particle $\cal P$ with
coordinate $\bm x $ to be
measured, and at least one more particle $\cal M$, for the measurement
apparatus, with position $\bm y$. In reality, of course, a measurement
apparatus or even a quantum probe that serves as intermediate-scale
device from the quantum world to the macroscopic world, would normally
be much more complex than a single particle.  But clearly, at least one
more particle is needed and it turns out to be enough for the
argument.  Let us consider a very crude measurement of the
where-abouts of the particle: we only distinguish between the particle
being in a region $\cal A$ (e.g.~the region inside the trap), and the
rest of the world $\cal B=\cal U\setminus \cal A$, where $\cal U$ is the
entire accessible universe.  Measuring whether the particle is still
anywhere in $\cal U$ is then equivalent  to assessing if it is in $\cal
A$ or in $\cal B$, with only three logical possible outcomes,
including that it is found neither in $\cal A$ nor in $\cal B$.  Now, if we do that measurement in Our World, by
definition the measurement apparatus (or detector) $\cal M$ is in Our World,
i.e.~$\bm y^{\rI,\rII}\in \cal U$. In fact, since $\bm y\in \cal A$ or $\bm y \in \cal
B$, we have necessarily that either $(\bm y^{\rI}\in {\cal A}) \land  (\bm y^{\rII}\in \cal
A)$ or $(\bm y^{\rI}\in {\cal B}) \land  (\bm y^{\rII}\in \cal B)$. For the sake of
the argument it is enough to assume that $\cal M$ be
localized, i.e.~we can assume that we know that $(\bm y^{\rI}\in {\cal A}) \land  (\bm y^{\rII}\in \cal
A)$. In order to measure anything, we clearly need an interaction between $\cal P$
and $\cal M$, and since we have two identical Twin Worlds, the interaction
must happen independently in each Twin World.  We can code the
different responses of $\cal M$ depending on whether $\bm x,\bm y$
are in $\cal A$ or $\cal B$ by different values $V^{AA}$, $V^{AB}$,
$V^{BA}$, or $V^{BB}$ of the interaction potential.  We need not model
the response of the detector to these different interactions, but could
imagine e.g.~that particle $\cal M$ moves to a different position
depending on whether the interaction is $V^{AA}$ or $V^{BA}$.  In any
case, different interactions are the 
first step in a hierarchy of information processing leading to
different outcomes of a more realistic detector depending on the relative positions
of $\cal P$ and $\cal M$.   More precisely,
in Twin World I, we have e.g.~$V^{AB,I}=V(\bm x^{\rI}\in {\cal A} \land
\bm y^{\rI}\in \cal B)$ and in   Twin World II, we have
e.g.~$V^{AB,II}=V(\bm x^{\rII}\in {\cal A} \land
\bm y^{\rII}\in \cal B)$. If $\cal M$ is to be of any use in a position
measurement of $\cal P$, we also need that $V^{AA}\ne V^{BA}$.  The
functional dependence of $V$ on its arguments 
must be the same in both Twin Worlds.  Now under the premises that we only see coincidence
of detector responses in Our World, it follows from $y^{\rI},y^{\rII}\in
\cal A$ that either $\bm x^{\rI},\bm x^{\rII}\in \cal A$ as well or
$\bm x^{\rI},\bm x^{\rII}\in \cal B$, whereas e.g.~$\bm x^{\rI}\in \cal A$ and
$\bm x^{\rII}\in \cal B$ would lead to two different detector responses in
the two Twin Worlds and hence is never observed in Our World.  In
other words, the assumed localization of the detector $\cal M$ in Our
World together with the assumption that only coincidence events in Our
World are observed leads to a post-selection bias that makes 
higher dimensions (i.e.~$ \bm x^{\rI}\ne \bm x^{\rII}$) is never observed. \\

More generally, if we have a statistical theory that reproduces
exactly the statistics of standard
quantum theory for all possible events, the questions asked in the framework of the new theory
should also be asked in standard quantum theory --- and the question whether
a particle (or the Moon) is there when nobody looks goes all the way
back to the beginnings of quantum mechanics, without satisfactory
answer still today.  In the Twin World picture we get the answer that
by construction we (or the measurement apparatus) do not even
measure the particle when it is not in Our World, so all measurement
records consist by definition only of events in which the particle is
present in Our World.

\subsection{Reflections on non-locality}\label{sec.nonloc}
The agreement of the Twin World picture of the quantum world with standard QM
implies that it is not possible to rewrite it as a local realistic
theory as modeled by Bell and CHSH. The Twin World picture is explicitly
realistic, in the sense that all variables have a well-defined value
at all times, which is different from the standard interpretations of
QM, where those values are miraculously created during the measurement
process and only then come into reality.  In contrast, in the Twin
World picture, we are just ignorant of those microscopic values and
hence describe them with probability distributions ---
the picture that Einstein kept advocating, but which was
abandoned based on Bell's definition of a local realistic theory and
overwhelming experimental evidence that his model does not describe Nature.  

If one tried to describe the Twin World picture with Bell's model of a
local realistic theory (see e.g.~\cite{PhysRevA.78.032114} for a modern formulation of the
Bell polytope), one would conclude that it cannot be local, as
it is explicitly realistic. This would be rather disturbing, as in
Bell's model it would imply signaling between 
possibly space-like separated players Alice and Bob.  
However, the statistical model of the Twin World picture is
different from Bell's.  While the latter appears to be perfectly
natural and tends to go unquestioned, we recognize that it is clearly
not the only possible statistical model. The possibility of
post-selection is not considered, nor the possibility of a refreshment
mechanism taking place in two Twin Worlds in parallel.  Independently
of the statistical model, one should ask therefore: 
Can Alice's measurement settings influence the probabilities of
outcomes on Bob's side?  As by construction the Twin World model
reproduces exactly the standard quantum statistics, and it is
generally accepted that quantum
statistics peacefully co-exists with special relativity
\cite{shimony1984controllable,Shimony1986-SHIEAP} (see, however,
\cite{peacock_peaceful_1991}), the answer is 
``No!'', on the same level as it is ``No!'' for standard quantum
mechanics. The non-local collapse of the wave function, considered as
problematic if the wave function is an element of reality, seizes to
be problematic as {\em in fine} it is just the collapse of a probability
distribution in either of the Twin Worlds. \\

Does the refresh mechanism lead to additional non-locality? At the
example of the CHSH-test, one can try to express the refreshments as
another stochastic map and examine whether it can be written as a
tensor product of two local stochastic maps.  If so, we can permit
that it depends locally on the measurement settings of Alice and Bob without
the need of information transfer between them.  This is analyzed in Appendix
\ref{app.refr}. The numerical results imply that the refreshments can
just marginally not be written
as local stochastic maps, but for a definite answer more precise methods or even proofs
are called for.  If confirmed, in this sense a new type of
non-locality arises.  However, if it exists, it acts in each of
the two Twin Worlds, i.e.~in two spaces different from our usual
configuration space, where we don't have 
experimental guidance on what is possible.  Non-Locality there does
evidently not lead to  non-local action or faster-than-light
information transfer in Our World, due to the agreement with the
statistical predictions of standard quantum mechanics in Our World and
the peaceful co-existence of standard quantum mechanics and special
relativity mentioned.

\subsection{Reduction of the quantum measurement problem to a
  classical one}
Under ``measurement problem'' one commonly understands the fact that
according to standard quantum mechanics, in a measurement the wave function collapses
to a random, single component of a superposition of eigenstates of an
observable represented by a hermitian operator.  The surviving
component is the  projection of the state immediately before
the collapse onto the eigenspace corresponding to the found
measurement result in von Neumann measurements, or corresponding to a
realized POVM element in the case of POVM (=positive operator-valued
measure) measurements.  This is problematic for several reasons: i.) It requires a cut
between the microsopic world and the macroscopic world with its
measurement instruments that is somewhat arbitrary; ii.) it leaves
unclear what exactly constitutes a measurement, and hence when the wave
function is supposed to collapse as opposed to evolving continuously
according to the Schr\"odinger equation \cite{bell_against_1990};
iii.) the collapse is non-local.

Arguably, issue i.) is not fully resolved by the Twin  World picture.  As seen
above, we assume that we as macroscopic observers live in Our World and can only
observe coincidence events there, whereas the dynamics of quantum particles
takes place in parallel in two Twin Worlds. It is not clear whether
this constitutes a cut between macrocsopic and microscopic world, or
whether also we observers venture into two parallel Twin Worlds, but
also only observe our own coincidences in Our World and hence don't
realize such vagabonding. In the same
context one might also object that coincidences of very many
particles become extra-ordinarily unlikely. However, as long as they
are non-zero and elementary events that constitute a coincidence are
well defined, this may not be a problem: The flow of time in Our World
will ultimately have to be defined by events in Our World, such that
long time-intervals of no coincidences drop out.

However, issues ii.) and iii.) are reduced simply to their classical
counterparts.  And compared to QM, the non-locality of the
collapse of a probability distribution
is comparatively trivial and and does not give rise to worries about ``spooky action
at a distance'' \cite{einstein_born-einstein_1971} (p.158):  
If the wavefunction just represents some observer's 
information about the system, the collapse is on the same
level as the update of a probability distribution upon sampling from
it.  It represents statistical information --- or lack
thereof, i.e.~our expectation what will happen if we repeat the
experment many times, rather than the state of a single system --- and
its update upon smapling is 
unproblematic.  If, however, one postulates that the quantum state is all we can know about a 
single system and attributes physical reality to the wave function and
its collapse, the latter becomes deeply problematic for the mentioned
reasons.

Einstein remained convinced that quantum mechanics is
incomplete and only a statistical theory on top of a realistic
microscopic world, where physical observables have well-defined
values, revealed in experiment \cite{Einstein35}.  That a local realistic theory is
unviable has become mainstream credo in physics after the derivation
of Bell's inequality and generalizations thereof based on the
assumption of local realism and overwhelming experimental evidence
that these inequalities are violated.  This creates substantial difficulties,
however. Besides the measurement
problem, the different ontologies in quantum mechanics (with quantum
states in Hilbert space or,
equivalently, in the Heisenberg picture, hermitian operators as basic dynamical quantities) and general
relativity (with worldlines of particles in a dynamical space-time) are
one of the main obstacles in unifying quantum mechanics and general
relativity.

The Twin World picture, however, makes clear that Bell's model
of a local realistic theory is, while plausible, not the most general
one.   Post-selection of coincidence events in Our World, and
refreshment mechanisms that align probability distributions in each
Twin World with absolute values of the quantum mechanical wave
function are not accounted for.  These ingredients allow for a theory that has
e.g.~particle positions in each Twin World as realistic variables, and
yet 
is local in the sense of no non-local information transfer in Our
World between Alice and Bob based on the probability distribution of
the realistic variables.  It reproduces exactly the statistical
predictions of quantum mechanics and hence, of course, violates
Bell inequalities exactly as quantum mechanics does. In this sense it
must be non-local because it is realistic.  But the
non-locality of the 
collapse of the wave function is of no concern, just as the non-local collapse
of the probability distribution of classically correlated particles
upon measuring one of them is of no concern.

The refreshment mechanism introduces a new type of possible non-locality in the sense explained
in Section \ref{sec.nonloc}. However, the non-locality happens in
each Twin World as part of the dynamics, as far as we can tell from
the limited evidence of simulating the CHSH experiment, and the
believed necessity of locality of physical influences in Our World does not rule out
non-locality of refreshments in each Twin World.  Theoretical attempts
to account for the non-locality of quantum mechanics and the
propagation of quantum information, such as in
Bohmian mechanics, or coupling
to tachyonic fields \cite{feinberg_possibility_1967}, have led previously to speculations
of propagation of quantum information in another space, and experiments
that lower bounded the apparent speed of propagation of quantum
information to several orders of magnitude larger than the speed of
light in vacuum \cite{scarani_speed_2000,salart_testing_2008}. \\

To summarize, with the Twin World picture of reality in the quantum world, the
measurement process is at least 
reduced to the same questions in classical statistics.  E.g.~we describe the
diffusive Brownian process of a particle with a continuously evolving
probability distribution of its position.  We use a probability
distribution because of our ignorance of the actual position of the
particle until we measure it --- at which point our ignorance is
resolved and correspondingly the probability distribution collapses
``instantaneously'' (but really on the time scale on which we acquire the
information) to a delta-peak on the found position.  Here it is
assumed that we can separate the world into observed particle and
observer, and this is not considered problematic just because the
particle has a well-defined position at all times, whether we measure
it or not.  The same reasoning applies to the Twin World picture,
where the particle has well-defined coordinates in each Twin World
that undergo a stochastic jump process and only come to existence in
Our World upon coincidence in the two Twin Worlds. If there they are
revealed by measurement to an experimentalist, the probability
distribution with which she describes the particle collapses
``instantaneously'', whereas someone else without that information might still hang on to the
continuously evolving dynamics of the probability distribution.

\subsection{Random permutations as fundamental building blocks of
  quantum dynamics}
Due to the block structure \eqref{eq:Sblock} of the stochastic matrix
$S$ in \eqref{eq:Sfull}, $S$ is a bistochastic
matrix.  Hence, by Birkhoff's theorem
\cite{birkhoff_three_1946}, $S$ can be represented as a convex combination of
permutations.
The refreshments are based on the estimate $\hat{\phi}$ from relative 
frequencies, see eq.\eqref{eq:hatpsiN}. Also these refreshments can be realized as permutations,
albeit permutations that depend through the state estimate on the
actual configurations realized in the ensemble. Hence, with propagation and
refreshments taken together, we can
understand quantum dynamics as arising fundamentally from random
permutations. This adds another twist to the proposal by Penrose that
``Perhaps the most basic ingredients of nature are not things, but
combinatorial relations --- the way entities can be permuted or
interchanged.'', even though he appears to have quantum processes at
the Planck scale in mind \cite{Penrose89} (chapter 8-9). 

\subsection{Challenges of the continuum limit}
Physical theories, both quantum and classical, traditionally assign
continuously varying real numbers to the positions of particles, and
the Laplace operator with its second derivatives with respect to
positions in Schr\"odinger's equation requires these continuously
varying numbers.  In reality, of course, this is an
idealization.  No particle can ever be localized with infinite
precision, as this would require an infinite amount of energy.  Even
fundamental constants such as the speed of light, can, as matter of
principle, not be defined with an infinite number of digits
\cite{braun_intrinsic_2017}, and it has been argued that at the latest at
the Planck scale space becomes somehow fuzzy \cite{Ford00} and geometry possibly
non-commutative
\cite{Snyder1947,DFR1994,DFR1995,Connes1994,ChamseddineConnes1997}. Recently,
it was pointed out that less significant 
digits in the numerical value of an observable might play the role of
hidden variables \cite{Gisin2019Indeterminism,Gisin2020RealNumbers} in
classical mechanics and be truely random as far as concerns physics.
Nevertheless, in non-relativistic QM the assumption
of continuously varying positions works well, and one should therefore consider
whether this is also possible in the Twin World picture.  One
encounters several problems, however.  Firstly, in the limit
of vanishing lattice constant $a$ (set to 1 so far), the negative and
positive parts of the Laplacian, $\mathbb{1}/a^2$ and $O_1/a^2$ (see
eq.\eqref{eq:DelO1}) diverge, whereas their sum remains a
well-defined differential operator. This might not be too problematic,
however, as the two parts always act together in the generator
$G_T$. Secondly, we made the assumption of finite maximum potential
energy, i.e.~$\text{max}_{x}|\tilde{V}(x)|<\infty$. In the continuum
limit, however, a Coulomb potential, say, does diverge for distance
between the charges approaching 0. For standard QM this is not
problematic.  E.g.~the energy eigenstates and eigenvalues of the
hydrogen atom can be obtained to very good approximation without
taking into account the finite size of the proton.  But for the amplitude reductions built into $G_V$ it
is. In reality, distances between particles can certainly not go to
exactly 0, already for the reason that otherwise the position of one
particle relative to the other would be specified again with infinite precision.  When 
reaching very small distances and correspondingly large 
energies, standard QM needs to be replaced by a relativistic field
theory, and therefore a similar extension of the Twin World picture is
expected to be necessary if $a$ becomes very small.  But thirdly,
infinitely small $a$ is problematic because it leads to vanishing
joint probabilities.  Indeed, with $a\to 0$, probabilities in Our
World would be of measure zero in the Full World.  Also from this
perspective one expects therefore a minimal length scale $a$ as
necessary for the theory to work.  Since, on the other hand, we have
taken time as continuous, this will create certainly some tension when
trying to make the theory relativistic.  An obvious way out would
be to also discretize time, but this is beyond the scope of this work. 
In the same context, it is to be expected that coincidences need only
happen with a certain tolerance, $x^{I}\in \{x^{II}-\delta
x/2,x^{II}+\delta x/2\}$, in order to be registered in Our World,
because exact coincidence would again lead to probabilities of
measure zero in the Full World.  All
of this hints at the existence of a fundamental small length scale
$\delta x$, which needs to be microscopically small and relevant for
defining fundamental events in the sense of probability theory. 

\subsection{Relationship to other theories}
\subsubsection{Time-symmetric formulation of quantum mechanics}
The idea of post-selection was explored previously in the context of the time-symmetric formulation
of quantum mechanics, where state preparation and
measurement are treated on an equal footing \cite{aharonov_time_1964}.
This has led to the formalism of weak values \cite{aharonov_how_1988}, which in turn
was later shown to select position as the hidden variable in Bohmian
mechanics if one defines velocity via a weak value
\cite{wiseman_grounding_2007}.  There appears to be some similarity
between the Twin Worlds and the forward- and backward propagated
branch in the time-symmetric formulation.  In the latter, however,
the two wave functions are not constructed from differences of
classical probability distributions, and there 
appears to be the need of an experimentalist who does the
post-selection.  With the post-selection per coincidence from the two
Twin Worlds on the other hand, there is a natural post-selection
mechanism without the need of intervention of an
experimentalist or backward propagation.

\subsubsection{Nelson's approach}
Schr\"odinger's equation was derived by Nelson from a stochastic process
\cite{nelson_derivation_1966}, based on  the ansatz
$\Psi=e^{R+ iS}$.  Also there, the particle
position functions as realistic variable that undergoes a
Brownian-like stochastic process, making the theory a
non-deterministic, realistic theory.  
It was criticized on several accounts, such as requiring a uniform
background field that introduces the noise responsible for the
stochastic process, which in addition should be highly correlated in
situations where entanglement is generated. Nelson himself saw satisfying locality as
problematic in his approach \cite{nelson_quantum_1985}. 
Relativistic generalizations  were nevertheless 
attempted over the years, see
e.g.~\cite{PhysRevD.31.2521,PhysRevD.27.2912,PhysRevD.32.1375,pavon_stochastic_2001,fritsche_stochastic_2009}. 
Another criticism, raised by Wallstrom
\cite{wallstrom_derivation_1989}, is that the stochastic formulation
of Schr\"odinger's equation and Schr\"odinger's equation itself are only
fully equivalent if the single-valuedness of $e^{iS}$ is assumed, for
which he did not see a natural reason. One
might debate the justification of this criticism, as single-valuedness
of the wave-function in addition to the evolution according to
Schr\"odinger's equation also has to be assumed in standard QM for
deriving, e.g., the correct quantization of angular
momenta.  In any
case, however, there appears to be no such ambiguity in the Twin World
formalism, as the basic underlying quantity is a semi-positive definite probability distribution,
and the (realified) wave-function arises as its derivative which hence
has also unique sign, leading to a unique, single-valued complexified wave-function
at any point in space based on it. 
Future work will have to show whether a relativistic extension is possible. \\
On a more general note, recent developments in the construction of
a ``postquantum theory of classical gravity'', where gravity is kept
classical but subject to a sufficient amount of noise \cite{PhysRevX.13.041040}, might renew the
attractiveness of stochastic quantum mechanics and
stochastic quantum field theory \cite{guerra_structural_1981} by
providing a natural unifying framework based on noisy classical
dynamics for both quantum mechanics and gravity.

\acknowledgements{I thank Jakub Czartowski for a critical reading of
  the manuscript, and Anagha Shriharsha, Robert Tumulka, and Thibault Wachtelaer for discussions.}

\section{Appendix: non-locality of refreshments}\label{app.refr}
Borrowing from the nomenclature of entanglement theory, let us say that a
stochastic map $S$ acting on a probability distribution 
of random variable of a bipartite system is local if it can be
written as a convex combination of tensor products of local stochastic
maps $S^A_i$, $S^B_i$,
\begin{equation}
  \label{eq:Sloc}
  S=\sum_{i=1}^k p_i S_i^A\otimes S^B_i\,,
\end{equation}
with $\sum_i p_i=1$, and 
$k$ can be arbitrarily large \footnote{I do not
  call it separable, as separable stochastic maps are already defined
  differently in the mathematical literature, see \cite{enwiki:1319001334}.}.   If $k=1$, we say that the
stochastic map is a (tensor) product map. We can restrict $0< p_i\le
1$, as terms with $p_i=0$ do not contribute in eq.\eqref{eq:Sloc}.  
One might wonder whether not all stochastic maps on bipartite systems
can be  written in that way, but the example of a probabilistic 
swap gate on two bits shows that this is not the case:
\begin{lemma}
Let
\begin{equation}
  \label{eq:Sswap}
  S_\text{swap}=
  \begin{pmatrix}
    1& 0 & & 0 \\
    0& 1-p & p & 0\\
    0& p & 1-p& 0\\
    0& 0 & 0 & 1
  \end{pmatrix}
\end{equation}
Then it is not possible to write $S_\text{swap}$ in the form \eqref{eq:Sloc}.
\end{lemma}
\begin{proof}
The proof is straight forward, based on making an ansatz of the form \eqref{eq:Sloc},
and then showing that
it leads to an algebraic contradiction.  Indeed, define
\begin{equation}
  \label{eq:SA}
  S_{i}^A=
  \begin{pmatrix}
    a_0^{(i)}& a_1^{(i)}\\
    1-a_0^{(i)}&1-a_1^{(i)}
  \end{pmatrix},\,\,\,
    S_{i}^B=
  \begin{pmatrix}
    b_0^{(i)}& b_1^{(i)}\\
    1-b_0^{(i)}&1-b_1^{(i)}
  \end{pmatrix}
\end{equation}
All free variables $a_0^{(i)},\ldots, b_1^{(i)}$ satisfy $0\le a_0^{(i)}\le 1,\ldots,0\le b_1^{(i)}\le 1$.  
The first row of eq.\eqref{eq:Sloc} with $S=S_\text{swap}$ leads to the constraints
\begin{eqnarray}
  \sum_i p_i a_0^{(i)} b_0^{(i)}&=& 1\label{eq:constr8.2}\\
  \sum_i p_i a_0^{(i)} b_1^{(i)}&=& 0\label{eq:constr8.3}\\
  \sum_i p_i a_1^{(i)} b_0^{(i)}&=& 0\label{eq:constr8.4}\\
  \sum_i p_i a_1^{(i)} b_1^{(i)}&=& 0\label{eq:constr8.5}\,.
\end{eqnarray}
From \eqref{eq:constr8.2} we conclude that $a_0^{(i)}\ne 0$ and
$b_0^{(i)}\ne 0$ for all $i$, because if for some $i=j$ one of them
vanished, the remaining terms in \eqref{eq:constr8.2} could not sum to
1 anymore: $\sum_{i\ne j}p_ia_0^{(i)} b_0^{(i)}\le \sum_{i\ne
  j}p_i<1$. From  \eqref{eq:constr8.4}
we then get $a_1^{(i)}=0$ for all $i$. But on
the other hand, from the 01-10 element of $S_\text{Swap}$, we need
$1=\sum_{i}p_ia_1^{(i)}(1- b_0^{(i)})$, but the r.h.s. vanishes,
a direct algebraic contradiction.
\end{proof}
Next consider the refreshments needed for the four different settings
of the CHSH experiment, see FIG.~\ref{fig:CHSHgrabit}. In each of the two Twin
Worlds, they must
make the probability distribution after the last two $S_H$ gates
interference-free.  With four grabits involved in each Twin World, $S$
is a 256$\times$256 dimensional stochastic matrix. All $S^A_i$ and
$S^B_i$ are 16$\times$16 dimensional stochastic matrices with 240 free
positive variables, $0\le (S^A_{i})_{kl}\le 1$,   $0\le (S^B_{i})_{kl}\le 1$
with normalization of each column to 1-norm 1.\\

We limit ourselves to $k=1$.  It corresponds to the case
where neither information transfer nor non-local action of Alice and
Bob in each Twin World are needed. Satisfying \eqref{eq:Sloc} with $k>1$ would need
information transfer but would not require non-local action in the Twin
Worlds, whereas for refreshments that do not satisfy \eqref{eq:Sloc},
refreshments act non-locally in the Twin Worlds. We also note that
reducing refreshments to stochastic maps is a simplification that
cannot work in general, as refreshments are intrinsically non-linear,
i.e.~depend on the state themselves.  Here we can reformulate them as
stochastic map, as we know the input probability distribution $P'$ and
output probability distribution $P''$, or
equivalently, the sequence of gates that were applied starting from a
known input state.

Eq.\eqref{eq:Sloc} for $k=1$ leads then to a system $P''=(S^A_1\otimes
S^B_1 )P'$ of
$256$ quadratic equations for 480 variables, to be solved
under the constraints that both $S^A_1\asdef S^A$ and $S^B_1\asdef S^B$ are stochastic
maps. Given it directly to Mathematica, it returns an empty set of
solutions. In order to corroborate this, one can simplify the problem by realizing
that both $P'$ and $P''$ contain a lot of zeros.  Those in $P''$
imply that the entire corresponding row in $S$ must be equal to zero
with the possible exception of entries that multiply zeros in
$P'$. Those in $P'$ make that an entire corresponding column of $S$
does not contribute, such that with the exception of one entry equal 1
for respecting normalization, we can put all entries in the column to zero.  More formally,  
Define
${\cal N'}\defas \{(k,l)|P'_{(k,l)}=0\}$, ${\cal N''}\defas
\{(i,j)|P''_{(i,j)}=0\}$.   Then do the following:
\begin{enumerate}
\item Define local matrices 16$\times$16 $S^A$, $S^B$ with matrix
  elements $a_{ik}$, $b_{jl}$, respectively.   
\item Set $S_{ij,kl}=a_{ik}b_{jl}$ for all $(i,j),(k,l)$.
\item Set $S_{ij,kl}=0$ for all $(i,j)\in \mathcal{N}^{\prime\prime}$.
\item Set $S_{ij,kl}=\delta_{ij,kl}$ for all $(k,l)\in \mathcal{N}^{\prime}$. This
  may override some 0's from the previous step, but it does not
  matter, as the corresponding columns do not contribute.  So this
  step is just for assuring the correct normalization in columns
  otherwise identical to zero.  
\item In order to ensure normalization of $S$, replace all entries in the last non-zero
  row in a column with 1 minus the sum of the entries in the same column,
  i.e.~for all $(k,l)$ and with $J_\text{max}\defas
  \text{argmax}_{(i,j)\notin \mathcal{N}^{\prime\prime}}$
  \begin{equation}
    \label{eq:Sijren}
    S_{J_\text{max},kl}\mapsto 1-\sum_{(i,j)=1}^{J_\text{max}-1}S_{ij,kl}\,.
  \end{equation}
\item From the resulting $S$, collect all remaining variables
  $a_{ik}$, $b_{jl}$ and set all other variables in $S^A$ and $S^B$ to 0. 
\item In order to enforce the local constraints
  $\sum_{i=1}^{15}a_{ik}=1$,    $\sum_{j=1}^{15}b_{jl}=1$, replace the
  last non-zero elements  in $S^A$, $S^B$ by 1 minus the sum of the
  remaining elements above in each column, and remove those last
  non-zero elements from the list of free variables. For the first setting $(\theta_1,\theta_2)=(\pi/2,\pi/4)$ in the Bell experiment, this leaves 96 independent variables, 16 from $S^A$ and 80 from $S^B$.  
\item  Add the
  positivity constraint $a_{ik}, b_{jl}\ge 0$ for all remaining
  $a_{ik}, b_{jl}$, including the once replaced in step 7.
\end{enumerate}
Mathematica's (v.13.2) ``NSolve'' answers ``The solution set
contains a full-dimensional component; use Reduce for complete
solution information'', but then returns an empty set. With
``Reduce'', the calculation did not terminate even after several days
of calculation.

One can also try to solve the system by formulating it as a minization
problem of the
quartic form and sum-of-squares polynomial $Q=\sum_{I=1}^{256}(P''_I-(S^A_1\otimes S^B_1P')_I)^2$
over the remaining free and contributing variables under the mentioned
constraints.  Here, for the first
settings in the CHSH experiment, the minimum of $Q$  found by
``NMinimize'', is $Q\simeq 0.0107$ i.e.~it almost vanishes, but not quite.  Hence, it appears that 
the refreshments necessary in each Twin World for the grabit-probability distributions to be
interference-free at the end of the CHSH grabit circuit for the first
setting cannot be written as a tensor product of two local stochastic
maps, but further investigations would be needed to rigorously prove
and possibly generalize this statement to values $k>1$ (i.e.~local
rather than factoring stochastic maps) or other settings.

\bibliography{C:/Users/danyb/Documents/mypapers/bibfile_master/mybibs_bt}

\end{document}